\documentclass[a4paper,11pt,oneside,reqno]{amsart}
\usepackage[utf8]{inputenc}
\usepackage{amsmath,amsthm,amsfonts,latexsym,amssymb,bm,enumerate}
\usepackage{ae}
\usepackage{cite}
\usepackage{float}
\usepackage{lmodern}
\usepackage[T1]{fontenc}

\usepackage{color}
\usepackage[colorlinks,linktocpage]{hyperref}

\usepackage[dvips]{graphicx}
\usepackage{psfrag}
\DeclareGraphicsExtensions{.eps,.art,.ART,.ps}

\usepackage{rotating}

\newcounter{mnotecount}[section]

\newcommand{\s}{s}
\newcommand{\kk}{k_+}
\newcommand{\kkk}{k_-}

\usepackage[dvips]{graphicx}
\usepackage{psfrag}
\DeclareGraphicsExtensions{.eps,.art,.ART,.ps}

\DeclareFontFamily{OT1}{pzc}{}
\DeclareFontShape{OT1}{pzc}{m}{it}%
              {<-> s * pzcmi8t}{}
\DeclareMathAlphabet{\mathpzc}{OT1}{pzc}%
                                {m}{it}

\newtheorem{Thm}{Theorem}[section]

\newtheorem{Lem}[Thm]{Lemma}
\newtheorem{Cor}[Thm]{Corollary}
\newtheorem{Prop}[Thm]{Proposition}

\newtheorem{Hyp}[Thm]{Hypothesis}

\newtheorem{Rmk}[Thm]{Remark}

\newcommand{\cal}{\mathcal}

\renewcommand{\Psi}{\rho}

\newcommand{\mrpp}{\mathring}
\newcommand{\mv}{\mathring{v}}
\newcommand{\R}{\mathbb{R}}

\newcommand{\eps}{\varepsilon}

\newcommand{\truemu}{\frac{2\varpi}{r}-\frac{e^2}{r^2}+\frac{\Lambda}{3}r^2}
\newcommand{\mysigma}{{\cal P}}
\newcommand{\myumax}{U}

\newcommand{\ckrm}{\check r_-}
\newcommand{\ckrp}{\check r_+}
\newcommand{\hatz}{\widehat{\frac{\zeta}{\nu}}}

\newcommand{\cg}{\Gamma}
\newcommand{\vgr}{v_{\ckrm}}
\newcommand{\ugr}{u_{\ckrm}}
\newcommand{\ugrp}{u_{\ckrp}}
\newcommand{\vgrp}{v_{\ckrp}}
\newcommand{\ug}{u_r}

\newcommand{\vgs}{v_s}

\newcommand{\ckr}{\check{r}}
\newcommand{\uckr}{u_{\ckr}}

\newcommand{\gam}{\gamma}
\newcommand{\ugam}{u_{\gam}}
\newcommand{\vgam}{v_{\gam}}
\newcommand{\ruv}{{\cal R}_{(U,V)}}

\newcommand{\cD}{{\cal D}}

\newcommand{\Ric}{\operatorname{Ric}}

\begin{document}

\newcounter{enumii_saved}

\title[Occurrence of mass inflation with an exponential Price law]{On the occurrence of mass inflation for the Einstein-Maxwell-scalar field system with a cosmological constant and an exponential Price law}

\author{Jo\~ao L.~Costa}
\author{Pedro M.~Gir\~ao}
\author{Jos\'{e} Nat\'{a}rio}
\author{Jorge Drumond Silva}

\address{Jo\~ao L.~Costa: ISCTE - Instituto Universitário de Lisboa, Av.\ das Forças Armadas, 1649-026 Lisboa, Portugal and CAMGSD}
\email{jlca@iscte-iul.pt}
\address{Pedro M.~Gir\~ao, Jos\'{e} Nat\'{a}rio and Jorge Drumond Silva: CAMGSD, Instituto Superior T\'ecnico, Universidade de Lisboa, Av.\ Rovisco Pais, 1049-001 Lisboa, Portugal}
\email{pgirao@math.ist.utl.pt}
\email{jnatar@math.ist.utl.pt}
\email{jsilva@math.ist.utl.pt}

\subjclass[2010]{Primary 83C05; Secondary 35Q76, 83C22, 83C57, 83C75}
\keywords{Einstein equations, black holes, strong cosmic censorship, Cauchy horizon, scalar field, spherical symmetry}
\thanks{This work was partially supported by FCT/Portugal through UID/MAT/04459/2013 and grant (GPSEinstein) PTDC/MAT-ANA/1275/2014.}

\maketitle

\begin{center}
{\bf Abstract}
\end{center}

In this paper we study the spherically symmetric characteristic initial data problem for the Einstein-Maxwell-scalar
field system with a positive cosmological constant in the interior of a black hole, assuming an exponential Price law
along the event horizon. More precisely, we construct open sets of characteristic data which, on the outgoing initial null hypersurface  
(taken to be the event horizon), converges exponentially to a reference Reissner-N\"{o}rdstrom black hole at infinity.

We prove the stability of the radius function at the Cauchy horizon, and show that, depending on the decay rate of the initial data, mass inflation 
may or may not occur. In the latter case, we find that the solution can be extended across the Cauchy horizon with continuous metric and Christoffel symbols in $L^2_{{\rm loc}}$, 
thus violating the Christodoulou-Chru\'sciel version of strong cosmic censorship.

\setcounter{tocdepth}{1}
\tableofcontents

\vspace{5mm}

\section{Introduction}

\subsection{Strong cosmic censorship and spherical symmetry}
Determinism of a physical system, modeled mathematically by evolution equations, is embodied in the questions of existence and uniqueness of solutions for given initial data. The initial value problem (or Cauchy problem) is therefore the appropriate setting for studying these models. Well known examples of equations where the Cauchy problem is quintessential are those of Newtonian mechanics, the Euler and Navier-Stokes systems in hydrodynamics and Maxwell's equations of electromagnetism.

Historically, the geometric nature and mathematical complexity of the Einstein equations made it difficult to recognize that they also fit into this framework. It was not until the seminal work of Y.~Choquet-Bruhat \cite{Choquet52}, and her later work with R.~Geroch \cite{Choquet69}, that the central role of the Cauchy problem in general relativity was established. These results relied crucially on recognizing the hyperbolic character of the Einstein equations. Uniqueness of the solutions, as for any hyperbolic PDE, then follows from a domain of dependence property. The essence of \cite{Choquet69} consists precisely in showing that given initial data there exists a maximal globally hyperbolic development (MGHD) for the corresponding Cauchy problem, that is, a maximal spacetime where this domain of dependence property holds.

For the Einstein equations, global uniqueness fails, and therefore determinism breaks down, if extensions of MGHDs to strictly larger spacetimes can be found. The statement that generically, for suitable  Cauchy initial data,\footnote{We will not provide a full discussion of the subtleties regarding the formulation of this conjecture (see for instance \cite{ChruscielSCC, RingstromCauchy}). As an example, note that trivial extensions can occur simply because the initial data is given on an incomplete Cauchy surface.} the corresponding MGHD cannot be extended is known as the strong cosmic censorship conjecture (SCCC) \cite{PenroseSingularities, ChruscielSCC, ChristodoulouGlobalnew}.

A crucial point in the precise formulation of this conjecture is deciding what exactly is meant by an extension. Various proposals have been advanced, differing on the degree of regularity that is demanded for the larger spacetime. The strongest formulation would correspond to the impossibility of extending the MGHD with a continuous Lorentzian metric. This happens for instance in the Schwarzschild solution, where continuous extensions across the singularity $r=0$ do not exist \cite{Sbierski}. However, from the PDE point of view, the conjecture should rather prevent the existence of extensions which are themselves solutions of the Einstein equations, even in a weak sense. If we collectively represent the Christoffel symbols by $\Gamma$, then a weak formulation of the vacuum Einstein equations $\Ric=0$ can be represented schematically as
\[
\int_M \Ric \phi = 0 \Leftrightarrow \int_M (\partial \Gamma + \Gamma \Gamma) \phi = 0 \Leftrightarrow \int_M ( - \Gamma \partial \phi + \Gamma \Gamma \phi) = 0,
\]
for any test function $\phi$. Thus, as pointed out by Christodoulou and Chru\'sciel~\cite{Christodoulou:2008, ChruscielSCC}, any weak solution extension is ruled out if the Christoffel symbols fail to be in $L^2_{{\rm loc}}$. We shall therefore refer to the conjecture that extensions with Christoffel symbols in $L^2_{{\rm loc}}$ generically do not exist as the Christodoulou-Chru\'sciel formulation of strong cosmic censorship.

The reason why a genericity condition must be included in any version of the SCCC is that there exist well known MGHDs, arising from complete asymptotically flat initial data, which can be smoothly extended to strictly larger solutions of the Einstein equations; when such extensions exist, the boundary of the MGHD in the larger spacetime is known as the Cauchy horizon. The paradigmatic example exhibiting this behavior is the Kerr family of solutions, describing rotating black holes, where the Cauchy horizon occurs inside the event horizon. In \cite{PenroseBatelle} (see also \cite{SimpsonInternal}), Penrose provided a heuristic argument, based on the blue-shift effect, by which arbitrarily small perturbations of the black hole exterior would be infinitely amplified along the Cauchy horizon, turning it into a ``singularity'' beyond which no extension should exist. According to this picture, extensions of spacetimes across Cauchy horizons would be artifacts of very particular solutions, such as the Kerr family; therefore they should be unstable and devoid of physical significance.

An obvious path for studying the validity of the SCCC is then to focus on perturbations of these exceptional MGHDs where a Cauchy horizon is known to exist. Given the considerable difficulty of the full system of Einstein's field equations, it is natural to introduce symmetry assumptions to reduce the number of degrees of freedom, even though this necessarily implies some loss of genericity. Spherical symmetry is a popular choice, since it leads to partial differential equations in only two independent variables and is compatible with both the asymptotic flatness requirements for modelling isolated astrophysical systems and standard cosmological spacetimes. Moreover, spherically symmetric solutions with a Cauchy horizon, analogous to the Kerr family, can be obtained by including an electromagnetic field: these constitute the so-called Reissner-Nordstr\"{o}m family of solutions, describing charged black holes. In the words of John Wheeler, ``charge is a poor man's angular momentum''.

However, Birkhoff's theorem imposes local uniqueness of electrovacuum solutions in spherical symmetry, that is, it establishes that there are no gravitational dynamical degrees of freedom under these assumptions. For that reason, Christodoulou introduced in \cite{Christodoulou:1986} a model where the Einstein equations are coupled to a massless scalar field, arguably the simplest model which retains the wavelike behavior expected from the Einstein equations but is not constrained by Birkhoff's theorem. Another important consideration for choosing  a massless scalar field is that, unlike other matter models (e.g. perfect fluids), it does not develop singularities in the absence of gravity, and so any breakdown of the solution when coupled to the Einstein equations can be attributed to purely gravitational effects. 

By adding an electromagnetic field, Dafermos \cite{Dafermos1} adapted Christodoulou's model to study perturbations of the interior of a Reissner-Nordstr\"{o}m black hole, and in particular the stability of the Cauchy horizon. More precisely, he considered a characteristic initial value problem for the Einstein-Maxwell-scalar field system in spherical symmetry, with Reissner-Nordstr\"{o}m data on the event horizon and arbitrary data along an ingoing null hypersurface; the Reissner-Nordstr\"{o}m spacetime itself is, of course, a particular solution of this problem. This work established the following two results: first, the radius function does not vanish at the Cauchy horizon, and so the metric can always be extended continuously across it; second, if the Reissner-Nordstr\"{o}m initial data is sufficiently subextremal and the free data on the ingoing null hypersurface decays sufficiently slowly towards the event horizon then a scalar invariant, the so-called renormalized Hawking mass, blows up at the Cauchy horizon, a phenomenon known as mass inflation (first identified by Poisson and Israel \cite{IsraelPoisson}). This in turn leads to the blow-up of the Kretschmann scalar, and so no $C^2$ extensions are possible for the metric in this setting. In fact, something much stronger is expected, namely that mass inflation even prevents the existence of extensions with Christoffel symbols in $L^2_{{\rm loc}}$ (see \cite{DafermosBlack}).

These results were further extended by the authors in \cite{relIst1, relIst2, relIst3}, with the inclusion of a cosmological constant of any sign and a more detailed analysis of the solution at the Cauchy horizon, depending on the precise decay of the initial data. Notably, this work established that for sufficiently fast decaying initial data not only does mass inflation not occur, but also it is possible to extend the spacetime across the Cauchy horizon as a classical solution of the Einstein-Maxwell-scalar field system, thus calling into question the SCCC.

However, assuming Reissner-Nordstr\"{o}m data on the event horizon is a somewhat artificial problem. A more realistic model is obtained if gravitational collapse arises from the evolution of initial data prescribed on a Cauchy surface. In this case, following the heuristic work of Price in 1972 \cite{Price}, it is widely expected that, in the absence of a cosmological constant, the scalar field decays polynomially along the event horizon (with respect to an Eddington-Finkelstein-type null coordinate). The conjecture that in a generic gravitational collapse scenario the scalar field decays with some precise rate became known as {\em Price's law}. In \cite{Dafermos2}, Dafermos proved that the radius function does not vanish at the Cauchy horizon if a polynomial Price's law is assumed as an upper bound, while mass inflation occurs if the corresponding lower bound is also imposed. The validity of the polynomial Price's law as an upper bound was subsequently established by Dafermos and Rodnianski \cite{DafermosProof} for the spherically symmetric collapse of a massless scalar field, thus yielding stability of the Cauchy horizon for black hole formation in this setting. Nevertheless, the occurrence of mass inflation, and therefore inextedibility with Christoffel symbols in $L^2_{{\rm loc}}$, remains an open problem. In recent work, Luk and Oh \cite{LukOh1, LukOh2} showed that solutions resulting from gravitational collapse of generic asymptotically flat initial data satisfy 
an integral lower bound along the event horizon, which, although weaker than the
pointwise lower bound predicted by Price, turns out to be enough to rule out the
existence of $C^2$ extensions. Whether it suffices to establish mass inflation is
still not clear.

In the presence of a positive cosmological constant, it is widely expected that the corresponding Price law should guarantee exponential decay of the scalar field along the event horizon (see for instance the linear analysis in~\cite{DafermosWavedeSitter, Dyatlov}, the numerical study in \cite{Brady97} or the nonlinear stability results in \cite{HintzVasy}). In this paper, we will therefore consider such an exponential decay and extend the analysis in \cite{relIst1, relIst2, relIst3} to this case. We prove the stability of the radius function at the Cauchy horizon, and show that, depending on the decay rate of the initial data, mass inflation may or may not occur. In the latter case, we find that the solution can be extended across the Cauchy horizon with Christoffel symbols in $L^2_{{\rm loc}}$. A more precise statement of our results can be found in Theorem~\ref{thmMain}.

\subsection{Summary of the main results}

We consider the Einstein-Maxwell-real massless scalar field equations in the presence of a cosmological constant $\Lambda$ (in units for which $c=4\pi G=\varepsilon_0=1$):
\begin{align*}
& R_{\mu\nu} - \frac12 R g_{\mu\nu} + \Lambda g_{\mu\nu} = 2 T_{\mu\nu}; \\
& dF = d\star F = 0; \\
& \Box \phi = 0; \\
& T_{\mu\nu} = \partial_\mu \phi \, \partial_\nu \phi - \frac12 \partial_\alpha \phi \, \partial^\alpha \phi \, g_{\mu\nu} + F_{\mu\alpha} F_{\nu}^{\,\,\alpha} - \frac14 F_{\alpha\beta} F^{\alpha \beta} g_{\mu\nu}.
\end{align*}
These form a system of partial differential equations for the components of the spacetime metric $g$, the Faraday electromagnetic $2$-form $F$, and the real massless scalar field $\phi$; here $R_{\mu\nu}$ are the components of the Ricci tensor, $R$ is the scalar curvature, $\star$ is the Hodge star operator and $\Box$ is the d'Alembertian (all depending on $g$).

In the spherically symmetric case, we can write the metric in double null coordinates $(u,v)$ as 
\[
g=-\Omega^2(u,v)\,dudv+r^2(u,v)\, \sigma_{\mathbb{S}^2}, 
\]
where $\sigma_{\mathbb{S}^2} := d\theta^2 + \sin^2 \theta d\varphi^2$ is the round metric on the $2$-sphere $\mathbb{S}^2$. In this case, the Maxwell equations decouple from the system, since they can be immediately solved to yield
\[
F = - \frac{Q_e \, \Omega^2(u,v)}{2 \, r^2(u,v)} \, du \wedge dv + Q_m \sin \theta d\theta \wedge d\varphi.
\]
Here $Q_e$ and $Q_m$ are constants, corresponding to a total electric charge $4 \pi Q_e$ and a total magnetic charge $4 \pi Q_m$. The remaining equations depend only on the parameter 
$$e = \sqrt{{Q_e}^2+{Q_m}^2},$$
which we assume to be nonzero. They can then be written as follows (see~\cite{relIst1}):
a wave equation for $r$,
\begin{equation}\label{wave_r} 
\partial_u\partial_vr=-\frac{\Omega^2}{4r} - \frac{\partial_ur\,\partial_vr}{r} + \frac{\Omega^2e^2}{4r^3} + \frac{\Omega^2 \Lambda r}{4},
\end{equation}
a wave equation for $\phi$,
\begin{equation}\label{wave_phi} 
\partial_u\partial_v\phi=-\,\frac{\partial_ur\,\partial_v\phi+\partial_vr\,\partial_u\phi}{r},
\end{equation}
the Raychaudhuri equation in the $u$ direction,
\begin{equation}\label{r_uu} 
\partial_u\left(\frac{\partial_ur}{\Omega^2}\right)=-r\frac{(\partial_u\phi)^2}{\Omega^2},
\end{equation}
the Raychaudhuri equation in the $v$ direction,
\begin{equation}\label{r_vv} 
\partial_v\left(\frac{\partial_vr}{\Omega^2}\right)=-r\frac{(\partial_v\phi)^2}{\Omega^2},
\end{equation}
and a wave equation for $\ln\Omega$,
 \begin{equation}\label{wave_Omega} 
\partial_v\partial_u\ln\Omega=-\partial_u\phi\,\partial_v\phi-\,\frac{\Omega^2e^2}{2r^4}+\frac{\Omega^2}{4r^2}+\frac{\partial_ur\,\partial_vr}{r^2}.
\end{equation}

We summarize the  main results of this paper in the following statement.
\begin{Thm}
\label{thmMain}
Consider the characteristic initial value problem for the spherically symmetric Einstein-Maxwell-scalar field system \eqref{wave_r}-\eqref{wave_Omega} on the domain $[0,U] \times \left[0,\infty\right[$, written in null coordinates $(u,v)$ determined by the conditions $\partial_vr(0,v)=g(\nabla r, \nabla r)(0,v)$ and $\partial_ur(u,0)=-1$. Take any subextremal element of the Reissner-Nordstr\"{o}m family of solutions with mass $\varpi_0$, non-vanishing charge parameter $e$ and cosmological constant $\Lambda$, and let $r_+$, $k_+$ and $k_-$ be, respectively, the corresponding event horizon radius and the surface gravities of the event and the Cauchy horizons. Then, for any $\varepsilon > 0$ and  fixed $s>0$ it is possible to construct open sets of initial data such that $\phi(u,0)$ is free along the ingoing null direction $\{v=0\}$, while $r(0,v) \to r_+$, $\partial_vr(0,v) \to 0$ and
\begin{equation} \label{PriceThm}
e^{-(sk_+ + \varepsilon)v} \, \lesssim \, \partial_v\phi(0,v) \, \lesssim \, e^{-(sk_+ - \varepsilon)v}
\end{equation}
as $v \to \infty$ along the event horizon $\{u=0\}$. 

Given $U > 0$ sufficiently small, there exists a unique maximal development of this characteristic initial value problem, defined on a past set ${\cal P} \subset [0,U] \times [0,\infty[$.
Moreover, for small enough $\varepsilon>0$ the following results hold (with $\rho=k_-/k_+>1$):
\begin{enumerate}[{\rm (1)}]
\item  {\em Stability of the radius function at the Cauchy horizon (Theorem~\ref{stability_of_Cauchy_horizon}).} There exists $U>0$ such that
$$[0,U]\times[0,\infty[\,\subset {\cal P},$$
and $r_0>0$ for which
$$r(u,v)>r_0,\ {\rm for\ all}\ (u,v)\in[0,U]\times[0,\infty[.$$
Consequently, $({\cal M},g,\phi)$ extends, across the Cauchy horizon $\{v=\infty\}$, to $(\hat {\cal M},\hat g,\hat \phi)$, with $\hat g$ and $\hat \phi$ in~$C^0$.
\end{enumerate}
\begin{enumerate}[{\rm (1)}]\addtocounter{enumi}{+1}
\item {\em Mass inflation (Theorem~\ref{mass_inflation_thm}).} If $\s<\min\left\{\rho,2\right\}$ then the renormalized Hawking mass $\varpi$ (see \eqref{bar_rafaeli}) satisfies
\begin{equation*}
\lim_{v\rightarrow\infty}\varpi(u,v)=\infty,\ {\it for\ each}\/\  0<u\leq U.
\end{equation*}
In particular, no $C^2$ extensions across the Cauchy horizon exist.
\item {\em No mass inflation (Theorem~\ref{no_mass}).}  If $\rho<\frac{9}{7}$ and $s>\frac{14}{9} \rho$ then
\[
\lim_{v\rightarrow\infty}\varpi(u,v)<\infty,\ {\it for\ each}\/\  0<u\leq U,
\]
provided that $U$ is sufficiently small.
\item {\em Breakdown of the Christodoulou-Chru\'sciel criterion (Theorem~\ref{breakdown}).} 
Under the same hypotheses as in {\rm (3)}, the Christodoulou-Chru\'sciel inextendibility criterion fails, i.e.\ 
$({\cal M},g,\phi)$ extends across the Cauchy horizon to $(\hat {\cal M},\hat g,\hat \phi)$, with $\hat g$ and $\hat \phi$ in~$C^0$,
Christoffel symbols $\hat\Gamma$ in $L^2_{\rm loc}$, and $\hat\phi$ in $H^1_{\rm loc}$.
\end{enumerate}
\end{Thm}
The regions of the $(\rho,s)$ plane where we can prove mass inflation and no mass inflation are depicted in the following figure.\footnote{The region where we can prove no mass inflation is not expected to be sharp, since the linear analysis carried out in~\cite{CostaFranzen, peter2} suggests that there exist $H^1$ extensions for $s>\rho$.}

\begin{center}
\begin{psfrags}
\psfrag{a}{{\hspace{-1.7cm} $_{(1,1)}$ \qquad $_{\frac{9}{7}}$}}
\psfrag{b}{{$_{2}$}}
\psfrag{c}{{\!\!\!\!$_{\frac{14}{9}}$}}
\psfrag{d}{{\!\!$_{2}$}}
\psfrag{r}{{\tiny $\rho$}}
\psfrag{s}{{\tiny $s$}}
\includegraphics[scale=.8]{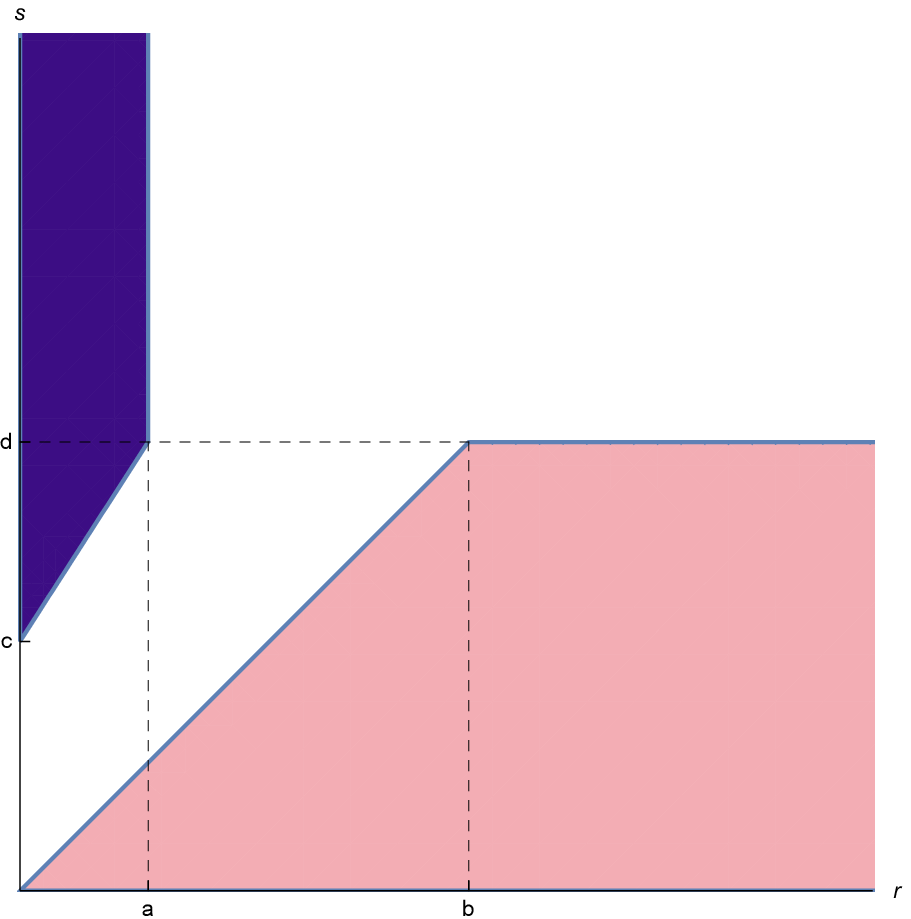}\end{psfrags} \label{figura}

\vspace{.5cm}\hspace{1cm}
\parbox{8cm}{
{\color[rgb]{0.24, 0.05, 0.52}\rule{.3cm}{.3cm}} - no mass inflation

{\color[rgb]{0.95, 0.68, 0.71}\rule{.3cm}{.3cm}} - mass inflation

}

\end{center}

\subsection{Implications for cosmic censorship}

As discussed above, the results in this paper do not apply directly to the SCCC, since this conjecture refers to  
global uniqueness of solutions arising from generic Cauchy data, while we consider characteristic data prescribed on 
a dynamic event horizon along which the scalar field satisfies a Price law of the form~\eqref{PriceThm}; 
that is, our results assume that a black hole is already present, as well as a specific decay of the field 
in its exterior.     

However, as is clear from Theorem~\ref{thmMain}, just the qualitative change in Price's law from polynomial ($\Lambda=0$) to exponential ($\Lambda>0$) is not
enough to obtain definitive conclusions about the behavior of the solutions at the Cauchy horizon, and therefore the validity of the SCCC. The final outcome requires,
in particular, a very precise quantitative knowledge of the value of $sk_+$ in~\eqref{PriceThm}, or, more precisely, of how such quantity
relates to the surface gravities of the Cauchy, event and cosmological horizons.
This is in stark contrast with the asymptotically flat case ($\Lambda=0$), where Price's law is expected to yield an inverse power
decay, which in turn is enough to establish mass inflation in the entire subextremal parameter range.

A recent numerical calculation of the quasinormal modes of Reissner-Nordstr\"om-de Sitter~\cite{cardosoRNdS} considerably changed the perspective on how $sk_+$ depends on the black hole parameters. In particular, this work (numerically) disproved a long-standing conjecture with roots in \cite{Brady97} and \cite{BradyCosmic}. Strictly speaking, the results in~\cite{cardosoRNdS} only apply to solutions of the linear wave equation in a fixed Reissner-Nordstr\"om-de Sitter background, but we expect similar results to also hold for the Einstein-Maxwell-scalar field system with a positive cosmological constant $\Lambda$. Assuming that this is the case, \cite{cardosoRNdS} shows that,  in the limit of large charge (for which $\rho$ approaches $1$), the decay corresponds to $sk_+$ close to $2k_-$ (that is, $s$ close to $2\rho$). In this regime, our results guarantee the existence of solutions with no mass inflation (see the figure above), as well as the existence of extensions beyond the MGHD with Christoffel symbols in $L^2_{loc}$. Note that the large charge limit can easily be obtained by picking a large cosmological horizon radius and then choosing the Cauchy horizon radius suitably close to the event horizon radius; these choices are in loose agreement with what one expects from the parameters of some astrophysical black holes.

It is interesting to note that the extendibility identified in~\cite{cardosoRNdS} occurs for near extremal black holes, where the blueshift effect is weaker. This feature can be compared with the fully nonlinear results in~\cite{GajicLuk}, where it is proved that one can indeed extend solutions of the Einstein-Maxwell-scalar field system across the Cauchy horizon of extremal black holes for $\Lambda=0$. The absence of blueshift in this case suggests that a similar result should hold for $\Lambda>0$. Similarly, one can expect the Cauchy horizon stability in the non-spherically symmetric setting, recently proved in~\cite{DafermosLuk} for $\Lambda = 0$, to remain true for $\Lambda>0$. Moreover, it is likely that these latter solutions can also be extended across the Cauchy horizon for near extremal  (i.e. rapidly rotating) black holes.

In conclusion, our results indicate that, with our current knowledge, the validity of the SCCC in the presence of a positive cosmological constant does not stand on firm ground. Nonetheless, the final verdict will only become clear once a precise quantitative understanding of the exponential Price law for $\Lambda>0$, in the full non-linear setting,  is achieved.

\subsection{Technical overview}

Introducing an exponential Price law creates new difficulties when compared to simply prescribing Reissner-Nordstr\"{o}m data along the event horizon, as in \cite{relIst2, relIst3}. We now summarize the main new technical features of the present work.

As for any characteristic initial value problem, our initial data is constrained by the evolution equations, and thus cannot be freely chosen (see Section~4 of \cite{relIst1}). Solving these nonlinear constraint equations while at the same time guaranteeing that our data describes a dynamical event horizon along which the scalar field decays at a prescribed rate is a non-trivial task; this problem is solved in Section~\ref{section3}. We use the radial derivative of the renormalized Hawking mass on the outgoing direction as the pivotal free function from which all the remaining quantities can then be constructed. In addition to exponentially decaying initial data, we produce sets of initial data with different types of decay, including the polynomial case studied by Dafermos \cite{Dafermos2}.

An important qualitative feature of the dynamics in the interior of the black hole is the celebrated redshift effect, characterized by an exponential decay of the form $e^{-2k_+v}$. This is the fastest decay that can be expected to be carried over by the evolution from the event horizon towards the Cauchy horizon. As exemplified in \cite{Dafermos2}, when the decay of the initial data is slower than exponential then it overwhelms the redshift effect and mass inflation always occurs. On the other hand, for faster than exponential decaying initial data the redshift effect dominates, and mass inflation may not occur, as was found in \cite{relIst3}. Therefore, initial data asymptotic to $e^{-sk_+v}$ constitutes the most interesting case, as it allows for a detailed analysis around the threshold value $s=2$, and is the only case that we will pursue.

Contrary to what happens in the Reissner-Nordstr\"{o}m solution, the interior of the black hole solutions that we are now considering does not coincide with the trapped region. In particular, an apparent horizon forms, whose asymptotic geometry must be understood, mostly by soft arguments, before proceeding to study the solution in greater detail (Section~\ref{section3.5}). To precisely estimate all relevant quantities in the region ${\cal P}_{\lambda}$ that lies in the past of the apparent horizon, as well as the region in its immediate future, we develop a two-dimensional version of Gr\"onwall's inequality adapted to this problem (Sections~\ref{section4} and \ref{section5}). As the Cauchy horizon is approached, the analysis becomes akin to that in~\cite{relIst2, relIst3} (Sections~\ref{section6} to \ref{section10}).

In this work we are not able to exclude the existence of non-trivial solutions whose radius function and renormalized Hawking mass are constant along the Cauchy horizon\footnote{In fact we conjecture that such solutions do exist.} (in ~\cite{relIst2, relIst3} such behavior was ruled out simply by assuming a non-zero ingoing perturbation). This creates new difficulties when analyzing solutions with no mass inflation and trying to construct extensions with Christoffel symbols in $L^2_{{\rm loc}}$ beyond the Cauchy horizon, which are averted by introducing a novel change of coordinates (see Section~\ref{section11}).   

\section{The spherically symmetric Einstein-Maxwell-scalar field \newline equations as a first order system}

To write the Einstein equations as a first order system of PDE we define the following quantities:
\begin{equation}\label{nu_0}
\nu:=\partial_u r,
\end{equation}
\begin{equation}\label{lambda_0}
\lambda:=\partial_v r,
\end{equation}
\begin{equation}\label{bar_rafaeli} 
\varpi:=\frac{e^2}{2r}+\frac{r}{2}-\frac{\Lambda}{6}r^3+\frac{2r}{\Omega^2}\nu\lambda,
\end{equation}
\begin{equation}\label{mu} 
\mu:=\truemu,
\end{equation}
\begin{equation}\label{theta} 
\theta:=r\partial_v\phi,
\end{equation}
\begin{equation}\label{zeta} 
\zeta:=r\partial_u\phi
\end{equation}
and
\begin{equation}\label{kappa_0} 
 \kappa:=-\frac{\Omega^2}{4\nu}.
\end{equation}
From~\eqref{bar_rafaeli} we obtain
\begin{equation}
 \lambda = \kappa (1-\mu).
\end{equation}
It is easy to see that
\[
1-\mu = g(\nabla r,\nabla r). 
\]
Therefore $\varpi$, like $r$, is a geometric quantity: it is called the renormalized Hawking mass.
Note that $(1-\mu)$ depends on $(u,v)$ only through $(r,\varpi)$. In what follows, in a slight abuse of notation, we will interchangeably regard $(1-\mu)$ as a function of either pair of variables, with the meaning being clear from the context.

The Einstein equations imply (see \cite{relIst1}) the following first order system for the variables $(r,\nu,\lambda,\varpi,\theta,\zeta,\kappa)$:
\begin{eqnarray} 
 \partial_ur&=&\nu\label{r_u},\\
 \partial_vr&=&\lambda\label{r_v},\\
 \partial_u\lambda&=&\nu\kappa\partial_r(1-\mu)\label{lambda_u},\\
 \partial_v\nu&=&\nu\kappa\partial_r(1-\mu),\label{nu_v}\\
 \partial_u\varpi&=&\frac 12(1-\mu)\left(\frac\zeta\nu\right)^2\nu,\label{omega_u}\\
 \partial_v\varpi&=&\frac 12\frac{\theta^2}{\kappa},\label{omega_v}\\
 \partial_u\theta&=&-\,\frac{\zeta\lambda}{r},\label{theta_u}\\
 \partial_v\zeta&=&-\,\frac{\theta\nu}{r},\label{zeta_v}\\
 \partial_u\kappa&=&\kappa\nu\frac 1r\left(\frac{\zeta}{\nu}\right)^2,\label{kappa_u}
\end{eqnarray}
with the restriction
\begin{equation}\label{kappa_at_u} 
\lambda=\kappa(1-\mu).
\end{equation}
From~(\ref{r_v}), (\ref{nu_v}), (\ref{omega_v}) and (\ref{kappa_at_u}) we obtain
\begin{equation}\label{ray_v_bis} 
\partial_v\left(\frac{\nu}{1-\mu}\right)=\frac{\nu}{1-\mu}\left(\frac{\theta}{\lambda}\right)^2\frac{\lambda}{r},
\end{equation}
which can also be written as
\begin{equation}\label{ray_v_dois} 
\partial_v\left(-\frac{\lambda}{\kappa\nu}\right)=\frac{\theta^2}{\kappa\nu r}.
\end{equation}
This is the Raychaudhury equation~\eqref{r_vv} written in these variables. The Raychaudhury equation~\eqref{r_uu} corresponds to~\eqref{kappa_u}.

Existence and uniqueness for the characteristic initial value problem associated to the first order system~\eqref{r_u}$-$\eqref{kappa_at_u}, as well as a continuation criterion, were studied in~\cite{relIst1}. There it was also shown that, under appropriate regularity conditions for the initial data, this system implies the Einstein equations. We shall therefore study the spherically symmetric Einstein-Maxwell-scalar field system in this framework.

\section{Initial conditions and behavior at the event horizon}\label{section3}

We wish to study the interior of a black hole of finite mass arising from gravitational collapse. In order to do that, we consider a coordinate system $(u,v)$ such that $u=0$ corresponds to the event horizon and $v$ increases along the outgoing null direction. To eliminate the remaining gauge freedom in the choice of coordinates we set
\begin{equation}\label{iu}
\begin{cases}
\nu(u,0)=-1,\\
\kappa(0,v)=1.
\end{cases}
\end{equation}
For this choice of $\kappa(0,v)$, geodesic completeness of the event horizon requires that the $v$ coordinate takes values in $\left[0,\infty\right[$. We assume that the coordinate $u$ takes values in  $\left[0,U\right]$, with $U>0$ to be chosen. 

Integration of equations~\eqref{nu_v} and~\eqref{kappa_u} with initial conditions \eqref{iu} implies $\kappa>0$ and $\nu<0$ over the whole solution domain. From~\eqref{ray_v_dois} we then have
\[
\partial_v\left(-\frac{\lambda}{\kappa\nu}\right) \leq 0.
\]
A simple consequence of the Mean Value Theorem yields the following result.
\begin{Lem}\label{sign}
Let $u\in[0,U]$.
\begin{enumerate}[{\rm(i)}]
\item
If $\lambda(u,\bar v)=0$ then $\lambda(u,v)\leq 0$ for all $v>\bar v$.
\item
If $\lambda(u,\bar v)<0$ then $\lambda(u,v)< 0$ for all $v>\bar v$.
\end{enumerate}
\end{Lem} 

Hawking's area theorem implies that $\lambda$ must be nonnegative over the event horizon.
In previous papers (see \cite{relIst1, relIst2, relIst3}) we considered the case where $\lambda(0,v)\equiv 0$.
The case where $\lambda$ starts out positive and then becomes identically zero can be reduced
to the one of the previous papers by fixing a new origin for the $v$ axis.
In view of the previous lemma, the only case that remains to be studied is the one where $\lambda$ is strictly positive over the event
horizon.

Under the previous hypotheses, equations \eqref{r_v} and \eqref{omega_v} imply that $r$ and $\varpi$ increase along the event horizon. To be consistent with the usual picture of gravitational collapse, we assume that the limits
\begin{equation}\label{rplus}
r(0,\infty)=r_+
\end{equation} 
and 
\begin{equation}\label{omegazero}
\varpi(0,\infty)=\varpi_0
\end{equation}
are finite,\footnote{In fact, this necessarily happens in the case $\Lambda>0$, as $\lambda=1-\mu$ must remain nonnegative along the event horizon, that is, $1-\frac{2\varpi}{r}+\frac{e^2}{r^2}\geq\frac{\Lambda}{3}r^2$.} which can be interpreted as data asymptotically converging to a Reissner-Nordstr\"{o}m black hole with (constant) renormalized Hawking mass $\varpi_0$ and (constant) event horizon radius $r_+$.
By the Mean Value Theorem, there exists a sequence $v_n\nearrow\infty$ such that $\lambda(0,v_n)\to 0$, and so, from \eqref{kappa_at_u} and \eqref{iu}, $(1-\mu)(0,v_n)\to 0$.
Equation~\eqref{mu}, together with the fact that $r(0,\,\cdot\,)$ and $\varpi(0,\,\cdot\,)$
have limits at infinity, implies $(1-\mu)(0,\,\cdot\,)$ has a limit at infinity.
We conclude that $r_+$ and $\varpi_0$ cannot be chosen arbitrarily, as they must satisfy
\begin{equation}\label{1mmu=0}
(1-\mu)(r_+,\varpi_0)=0.
\end{equation}
Moreover, we assume that the black hole is asymptotically non-extremal, that is,
$$
\kk:=\frac12\partial_r(1-\mu)(r_+,\varpi_0)>0.
$$
As is well known, this quantity is called the surface gravity of the event horizon for the Reissner-Nordstr\"{o}m black hole with parameters $r_+$ and $\varpi_0$.

As explained in \cite{relIst1}, for a given choice of coordinates the initial data for the characteristic initial value problem consists of two free functions, one along the ingoing null segment $v=0$ and the other along the event horizon $u=0$. On the ingoing null segment we can freely specify $\zeta(u,0)$, but on the event horizon the functions $\lambda(0,v)$, $\varpi(0,v)$ and $\theta(0,v)$ are interrelated through \eqref{omega_v} and \eqref{kappa_at_u}. Because of these constraints, as well as \eqref{rplus} and \eqref{omegazero}, it turns out that the simplest approach is to start by choosing $\varpi$ as a function of $r$ along the event horizon. The functions $\lambda(0,v)$ and $\theta(0,v)$ can then be obtained from \eqref{kappa_at_u} and \eqref{omega_v}, respectively, making sure in the end that $\lambda(0,v) > 0$. We will now describe this procedure in detail.

Since we assume that $\lambda(0,v)$ is strictly positive, $r(0,v)$ is a strictly increasing function of $v$, and so it may be used as a coordinate along the event horizon.
Accordingly, we will write a hat over a function to mean that it is written in terms of this new coordinate.

As explained above, we will start by specifying $\hat\varpi(r)$, but the restriction \eqref{omegazero} means that the true free function is its derivative $\hat \varpi'(r)$, which we will prescribe as a continuous integrable function $\hat f:\,]0,r_+[\,\to\R^+_0$, so that
\begin{equation}\label{omega_eh}
\hat{\varpi}(r)=\varpi_0-\int_r^{r_+}\hat f(\tilde r)\,d\tilde r.
\end{equation}

In terms of the $r$ coordinate, \eqref{iu} becomes
\begin{equation}\label{kappa_eh}
\hat\kappa(r)\equiv 1,
\end{equation}
from which~\eqref{kappa_at_u} implies
\begin{equation}\label{lambda_eh}
\hat{\lambda}(r)=\widehat{(1-\mu)}(r)=1-\frac{2\hat{\varpi}(r)}{r}+\frac{e^2}{r^2}-\frac{\Lambda}{3}r^2.
\end{equation}
From the discussion that leads to \eqref{1mmu=0} it is also clear that
\begin{equation} \label{lambdagoestozero}
\lim_{r\to r_+}\hat{\lambda}(r)=\lim_{r\to r_+}\widehat{(1-\mu)}(r)=0.
\end{equation}
Moreover,
$$
\hat{\lambda}'(r)=\partial_r(1-\mu)(r,\hat\varpi(r))-\frac{2\hat{\varpi}'(r)}{r}.
$$
We now make the extra assumption that\footnote{This is equivalent to assuming that $\lim_{v \to \infty} \frac{\theta^2}{\lambda}(0,v)=2A$.}
$$\lim_{r\to r_+}\hat f(r) = A.
$$
Then we have
\begin{equation}\label{lambda_d_eh}
\lim_{r\to r_+}\hat{\lambda}'(r)=\partial_r(1-\mu)(r_+,\varpi_0)-\frac{2A}{r_+} = \frac{2}{r_+}(r_+\kk - A).
\end{equation}
In view of \eqref{lambdagoestozero}, we see that a necessary condition for $\hat{\lambda}$ to be positive in a left
neighborhood of $r_+$ is
\begin{equation}
\label{A}
A\in\left[ r_+ \kk,\infty\right].
\end{equation}  

We define $\omega:\,]0,r_+]\to\R$ (not to be confused with $\varpi$) by
$$
\omega(r)=\frac r2\left(1+\frac{e^2}{r^2}-\,\frac{\Lambda}{3}r^2\right).
$$
This is the value of the mass that would make $(1-\mu)$ vanish at $r$,
\begin{equation}\label{n=1}
(1-\mu)(r,\omega(r))=0.
\end{equation} 
In particular $\omega(r)>\varpi(r)$ for $r<r_+$ and $\omega(r_+) = \varpi_0$.
We have
$$
\omega'(r)=\frac{1}{2}\left(1-\frac{e^2}{r^2}-\Lambda r^2\right),
$$
which can also be written, by differentiating \eqref{n=1}, as
\begin{equation}\label{omega-linha}
\omega'(r)=-\,\frac{\partial_r(1-\mu)(r,\omega(r))}
{\partial_\varpi(1-\mu)(r,\omega(r))}=\frac{r}{2}\partial_r(1-\mu)(r,\omega(r)).
\end{equation}
Moreover, it is easy to check that
$$
\omega(r)=\frac{r}{2}(1-\mu)(r,\varpi_0)+\varpi_0,
$$
and so, for $n\geq 1$,
\begin{equation}\label{particular}
\omega^{(n)}(r_+)=\left.
\partial^n_r\left(\frac{r}{2}(1-\mu)(r,\varpi_0)\right)\right|_{r=r_+}.
\end{equation}

In the case when $A=r_+\kk=\omega'(r_+)$,
we will also assume that
\begin{equation}
\label{f-A}
\hat{f}(r)>\frac{1}{2}\Bigl(1-\frac{e^2}{r^2}-\Lambda r^2\Bigr)=
\omega'(r)=
\frac{r}{2}\partial_r(1-\mu)(r,\omega(r))
\end{equation}
for $r$ in a left neighborhood of $r_+$.
We claim that this assumption
guarantees that $\hat\lambda$ is positive in a left neighborhood of $r_+$.
Indeed, from~\eqref{f-A}, we have, for $r<r_+$,
\begin{eqnarray*}
\hat\lambda'(r)&=&\widehat{(1-\mu)}'(r) \\
&=&\frac{2}{r^2}
\left(\varpi_0-\int_r^{r_+}\hat{f}(\tilde r)\,d\tilde r\right) 
-\,\frac{2}{r}\hat{f}(r)-\,\frac{2e^2}{r^3}-\,\frac{2\Lambda}3r\\
&<&\frac{2}{r^2}
\left(\varpi_0-\int_r^{r_+}\omega'(\tilde r)\,d\tilde r\right)
-\partial_r(1-\mu)(r,\omega(r))-\,\frac{2e^2}{r^3}-\,\frac{2\Lambda}3r\\
&=&\frac{2}{r^2}
\left(\frac{r}{2}+\frac{e^2}{2r}-\,\frac{\Lambda}{6}r^3\right) 
-\,\frac{1}{r}\left(1-\frac{e^2}{r^2}-\Lambda r^2\right)-\,\frac{2e^2}{r^3}-\,\frac{2\Lambda}3r\\
&=&0.
\end{eqnarray*}
This proves our claim.

In the case when $A>r_+\kk$,
 $\hat\lambda$ is positive in a left neighborhood of $r_+$ because
 $\lim_{r\to r_+}\hat\lambda'(r)<0$.

We now list the most relevant choices of $\hat f$ satisfying the hypotheses above, noting in particular that when $A=r_+\kk$ the function $\hat{f}$ is chosen to be the sum of a Taylor polynomial of $\omega'$ at $r_+$ with a term that ensures $\hat{f}>\omega'$ in a left neighborhood of $r_+$.

\begin{Hyp}[on $\hat f$]\label{hyp}
The function $\hat{f}:\,]0,r_+[\,\to\R^+_0$ is continuous, integrable and has
limit $A\in\left[r_+\kk,\infty\right]$ as $r \to r_+$.
In addition, in a left neighborhood of $r_+$ one of the five following alternatives holds:\footnote{We could work with other assumptions. 
But if $\hat{f}$ blows up too fast at $r_+$, for example,
$\hat{f}(r)\sim C(r_+-r)^\alpha$, with $-1<\alpha<0$, it can be proved that we are
led to an incomplete event horizon.}
\begin{enumerate}[{\rm (i)}]
\item $A=\infty$ and there exist $c,C>0$ such that
$$
-c\ln(r_+-r)\leq\hat{f}(r)\leq -C\ln(r_+-r).
$$
\item $r_+\kk<A<\infty$.
\item $A=r_+\kk$ and
there exist $c,C>0$ and $0<\alpha_1\leq\alpha_2<1$ such that
$$
\omega'(r_+)+c(r_+-r)^{\alpha_2}\leq
\hat{f}(r)\leq\omega'(r_+)+C(r_+-r)^{\alpha_1}.
$$
\item $A=r_+\kk$ and
there exist $c,C>0$ and $n\geq 1$ such that
\begin{eqnarray*}
&&\sum_{k=0}^n\frac{\omega^{(k+1)}}{k!}(r_+)(r-r_+)^k+c(r_+-r)^{n}\\
&&\qquad\qquad\qquad\qquad\qquad\qquad\leq\hat{f}(r)\leq\\
&&\qquad\qquad\qquad\qquad\qquad\qquad\sum_{k=0}^n
\frac{\omega^{(k+1)}}{k!}(r_+)(r-r_+)^k+C(r_+-r)^{n}.
\end{eqnarray*}
\item $A=r_+\kk$ and
there exist $c,C>0$, $n\geq 1$ and $n<\alpha_1\leq\alpha_2<n+1$ such that
\begin{eqnarray*}
&&\sum_{k=0}^n\frac{\omega^{(k+1)}}{k!}(r_+)(r-r_+)^k+c(r_+-r)^{\alpha_2}\\
&&\qquad\qquad\qquad\qquad\qquad\qquad\leq\hat{f}(r)\leq\\
&&\qquad\qquad\qquad\qquad\qquad\qquad\sum_{k=0}^n\frac{\omega^{(k+1)}}{k!}(r_+)(r-r_+)^k+C(r_+-r)^{\alpha_1}.
\end{eqnarray*}
\end{enumerate}
\end{Hyp}

Case~(iii) corresponds to~(v) with $n=0$. However, we consider case~(iii) separately because it is especially interesting, as it leads to the polynomial Price law studied by Dafermos in \cite{Dafermos2}. In fact, each case above yields a different type of Price law, as we will see in the remainder of this section. However, only case (ii) will be pursued in detail in the following sections, since it corresponds to the exponential Price law expected for a positive cosmological constant.

For each of the cases~(i) through~(iv) above, we now proceed to examine the behavior of the following functions along the event horizon:
\begin{enumerate}[{\bf (a)}]
\item $\hat{\lambda}$,
\item $\left|\frac{\hat{\theta}}{\hat{\lambda}}\right|$,
\item $r(0,\,\cdot\,)$,
\item $\lambda_0(\,\cdot\,):=\lambda(0,\,\cdot\,)$,
\item $\theta_0(\,\cdot\,):=\theta(0,\,\cdot\,)$.
\end{enumerate}
Case (v) is very similar to cases (iii) and (iv), and will not be treated explicitly.

\bigskip

\begin{center}
{\bf (a)} Estimates for $\hat{\lambda}$
\end{center}

The function $\hat{\lambda}$ is determined over the event horizon
using~\eqref{omega_eh} and~\eqref{lambda_eh}: 
\begin{equation}\label{quociente2}
\hat{\lambda}(r)=
\frac{2}{r}\left(
\frac{r}{2}(1-\mu)(r,\varpi_0)+\int_r^{r_+}\hat f(\tilde r)\,d\tilde r\right).
\end{equation}
In the cases where 
\begin{equation}
\label{obama}
\hat f(r)=\sum_{k=0}^n\frac{\omega^{(k+1)}}{k!}(r_+)(r-r_+)^k + \hat{e}(r),
\end{equation}
equation~\eqref{quociente2} can be written in the form
$$
\hat{\lambda}(r)=
\frac{2}{r}\left(
\frac{r}{2}(1-\mu)(r,\varpi_0)-\sum_{k=0}^n\frac{\omega^{(k+1)}}{(k+1)!}(r_+)(r-r_+)^{k+1}+\int_r^{r_+}\hat e(\tilde r)\,d\tilde r\right).
$$
Using~\eqref{particular}, we then have 
\begin{equation}
\label{quociente3}
\hat\lambda(r)=\frac{2}{r}\left(O((r_+-r)^{n+2})+
\int_r^{r_+}\hat e(\tilde r)\,d\tilde r\right)=(\widehat{1-\mu})(r).
\end{equation}

Let us denote by $\left]r_1,r_+\right[$ a left neighborhood of $r_+$ where one of the assumptions (i) through (iv) above holds and moreover $\hat{\lambda}=\widehat{(1-\mu)}$ is positive. In what follows the constants $c$ and $C$ will be as in Hypothesis~\ref{hyp}, and $\eps>0$ is a parameter that can be made arbitrarily small by choosing $r_1$ sufficiently close to $r_+$. 
\begin{enumerate}[(i)]
\item  
In this case, we note that
$$
\int_r^{r_+}-\ln(r_+-\tilde r)\,d\tilde r=-(r_+-r)\ln(r_+-r)+(r_+-r).
$$
So, using~\eqref{quociente2},
\begin{equation}\label{lambda-1}
-C_1(r_+-r)\ln(r_+-r)
\leq\hat{\lambda}(r)\leq -C_2(r_+-r)\ln(r_+-r),
\end{equation}
for $r \in \left]r_1,r_+\right[$. Here
$C_1=\frac{2c(1-\eps)}{r_+}$ and $C_2=\frac{2C(1+\eps)}{r_+}$.
\item This hypothesis implies
$$
\hat{f}(r)=A+o(1),
$$
as $r\to r_+$. Using~\eqref{quociente2}, we have
\begin{equation}\label{lambda-2}
C_1(r_+-r)\leq\hat \lambda(r)\leq C_2(r_+-r),
\end{equation}
for $r \in \left]r_1,r_+\right[$.
Here 
\begin{equation}\label{C1}
C_1=\left(\frac{2A}{r_+}-\partial_r(1-\mu)(r_+,\varpi_0)\right)(1-\eps)=\frac2{r_+}(A-r_+\kk)(1-\eps)
\end{equation}
and 
\begin{equation}\label{C2}
C_2=\left(\frac{2A}{r_+}-\partial_r(1-\mu)(r_+,\varpi_0)\right)(1+\eps)=\frac2{r_+}(A-r_+\kk)(1+\eps).
\end{equation}
\item In this case, using~\eqref{quociente3} with $n=0$, we obtain
\begin{equation}\label{lambda-3}
C_1(r_+-r)^{1+\alpha_2}
\leq\hat{\lambda}(r)\leq C_2(r_+-r)^{1+\alpha_1},
\end{equation}
for $r \in \left]r_1,r_+\right[$. Here
$C_1=\frac{2c(1-\eps)}{r_+(1+\alpha_2)}$ and $C_2=\frac{2C(1+\eps)}{r_+(1+\alpha_1)}$.
Recall that $0<\alpha_1\leq\alpha_2<1$.
\item Similarly to (iii), we have
\begin{equation}\label{lambda-6}
C_1(r_+-r)^{n+1}
\leq\hat{\lambda}(r)\leq C_2(r_+-r)^{n+1},
\end{equation}
for $r \in \left]r_1,r_+\right[$. Here
$C_1=\frac{2c(1-\eps)}{r_+(n+1)}$ and $C_2=\frac{2C(1+\eps)}{r_+(n+1)}$.
Recall that $n\geq 1$.
\end{enumerate}

\bigskip

\begin{center}
{\bf (b)} Estimates for $\bigl|\frac{\hat{\theta}}{\hat{\lambda}}\bigr|$
\end{center}

The quotient $\frac{\hat\theta}{\hat\lambda}$ 
is a continuous function such that
\begin{equation}\label{tl_eh}
\left(\frac{\hat\theta}{\hat\lambda}\right)^2=\frac{2\hat\varpi'}{\hat{\lambda}}=\frac{2\hat f}{\hat{\lambda}}
\end{equation}
(see \eqref{omega_v}), and so 
$$
\lim_{r\to r_+}\left(\frac{\hat\theta}{\hat\lambda}\right)^2(r)=\frac{2A}{0^+}=\infty.
$$
We now determine the rate of blow up of $\left|\frac{\hat\theta}{\hat\lambda}\right|$ at $r_+$ in each of the cases~(i) through (iv).

\begin{enumerate}[(i)]
\item 
Using~\eqref{tl_eh} and \eqref{lambda-1}, we get
$$
\frac{c_1}{(r_+-r)^{1/2}}\leq
\left|\frac{\hat\theta}{\hat\lambda}\right|(r)\leq
\frac{c_2}{(r_+-r)^{1/2}},
$$
for $r \in \left]r_1,r_+\right[$. Here
$c_1=\sqrt{\frac{2c}{C_2}}$ and $c_2=\sqrt{\frac{2C}{C_1}}$.
\item 
We have from~\eqref{tl_eh} and \eqref{lambda-2}
$$
\frac{c_1}{(r_+-r)^{1/2}}\leq
\left|\frac{\hat\theta}{\hat\lambda}\right|(r)
\leq\frac{c_2}{(r_+-r)^{1/2}},
$$
for $r \in \left]r_1,r_+\right[$. 
Here $c_1=\sqrt{\frac{2A}{C_2}}$ and $c_2=\sqrt{\frac{2A}{C_1}}$.
\item Using~\eqref{tl_eh} and \eqref{lambda-3},
we have
$$
\frac{c_1}{(r_+-r)^{\frac{1+\alpha_1}2}}
\leq\left|\frac{\hat\theta}{\hat\lambda}\right|(r)\leq
\frac{c_2}{(r_+-r)^{\frac{1+\alpha_2}2}},
$$
for $r \in \left]r_1,r_+\right[$. 
Here $c_1=\sqrt{\frac{2A}{C_2}}$ and $c_2=\sqrt{\frac{2A}{C_1}}$.
\item We have from~\eqref{tl_eh} and \eqref{lambda-6}
$$
\frac{c_1}{(r_+-r)^{\frac{n+1}2}}
\leq\left|\frac{\hat\theta}{\hat\lambda}\right|(r)\leq
\frac{c_2}{(r_+-r)^{\frac{n+1}2}},
$$
for $r$ in a left neighborhood of $r_+$. 
Here $c_1=\sqrt{\frac{2A}{C_2}}$ and $c_2=\sqrt{\frac{2A}{C_1}}$.
\end{enumerate}

\bigskip

\begin{center}
{\bf (c)} Estimates for $r(0,\,\cdot\,)$
\end{center}
Let us define
\begin{equation}
\label{r0_hat}
\hat r_0=\inf\{\hat r<r_+:\partial_r(1-\mu)(r,\hat\varpi(\hat r))>0\ {\rm for\ all}\ \hat r\leq r\leq r_+\}.
\end{equation}
(Note that $\hat r_0>r_0$, where $r_-<r_0<r_+$ is such that $\partial_r(1-\mu)(r_0,\varpi_0)=0$.)
Choose $r_2 \in\left]\max\{\hat r_0,r_1\},r_+\right[$.
To fix the $v$ coordinate we set $r(0,0)=r_2$, so that
\begin{equation}
\label{r-v}
v(r)=\int_{r_2}^r\frac{1}{\widehat{\partial_v r}(\tilde r)}\,d\tilde r
=\int_{r_2}^r\frac{1}{\hat{\lambda}(\tilde r)}\,d\tilde r.
\end{equation}
Whichever the case,
\eqref{lambda-1}, \eqref{lambda-2}, \eqref{lambda-3} or~\eqref{lambda-6},
 $1/\hat{\lambda}$
is not integrable in $]r_2,r_+[$, and so
$v(r_+)=\infty$.

Using~\eqref{r-v}, we can now determine the behavior of~$r$ as a function of~$v$ along the event horizon.
\begin{enumerate}[(i)]
\item 
In this case, we have
$$
(r_+-r_2)^{e^{C_2v}}\leq r_+-r(0,v)\leq(r_+-r_2)^{e^{C_1v}}.
$$
We can assume, without loss of generality, that $r_+-r_2<1$.
\item
Here, we obtain
$$
(r_+-r_2)e^{-C_2v}\leq r_+-r(0,v)\leq(r_+-r_2)e^{-C_1v}.
$$
\item
In this case, we have
$$
\frac 1{\left[\alpha_1C_2v+
\frac{1}{(r_+-r_2)^{\alpha_1}}\right]^{\frac{1}{\alpha_1}}}
\leq r_+-r(0,v)\leq
\frac 1{\left[\alpha_2C_1v+
\frac{1}{(r_+-r_2)^{\alpha_2}}\right]^{\frac{1}{\alpha_2}}}.
$$
Recall that $0<\alpha_1\leq\alpha_2<1$.
\item
Here, we obtain
$$
\frac 1{\left[nC_2v+
\frac{1}{(r_+-r_2)^{n}}\right]^{\frac{1}{n}}}
\leq r_+-r(0,v)\leq
\frac 1{\left[nC_1v+
\frac{1}{(r_+-r_2)^{n}}\right]^{\frac{1}{n}}}.
$$
Recall that $n\geq 1$.
\end{enumerate}

\bigskip

\begin{center}
{\bf (d)} Estimates for $\lambda_0$
\end{center}

The estimates for $r(0,v)$ obtained in~(c) now allow us to rewrite the bounds for $\hat{\lambda}$, determined in (a), as bounds for $\lambda_0$ in terms of~$v$.

\begin{enumerate}[(i)]
\item
In this case, we have
$$
C_1\ln\left({\textstyle\frac{1}{r_+-r_2}}\right)e^{C_1v}
\left({\textstyle\frac{1}{r_+-r_2}}\right)^{-e^{C_2v}}\leq 
\lambda_0(v)\leq C_2\ln\left({\textstyle\frac{1}{r_+-r_2}}\right)e^{C_2v}
\left({\textstyle\frac{1}{r_+-r_2}}\right)^{-e^{C_1v}}.
$$
Recall that $r_+-r_2<1$.
\item
Here, we obtain
\begin{equation}\label{l_ii}
C_1(r_+-r_2)e^{-C_2v}\leq\lambda_0(v)\leq C_2(r_+-r_2)e^{-C_1v}.
\end{equation}
\item
In this case, we have
\begin{equation}\label{l_iii}
\frac{C_1}{\left[\alpha_1C_2v+
\frac{1}{(r_+-r_2)^{\alpha_1}}\right]^{\frac{1+\alpha_2}{\alpha_1}}}
\leq
\lambda_0(v)\leq \frac{C_2}{\left[\alpha_2C_1v+
\frac{1}{(r_+-r_2)^{\alpha_2}}\right]^{\frac{1+\alpha_1}{\alpha_2}}}.
\end{equation}
Recall that $0<\alpha_1\leq\alpha_2<1$.
\item
Here, we obtain
$$
\frac{C_1}{\left[nC_2v+
\frac{1}{(r_+-r_2)^{n}}\right]^{\frac{n+1}{n}}}
\leq
\lambda_0(v)\leq \frac{C_2}{\left[nC_1v+
\frac{1}{(r_+-r_2)^{n}}\right]^{\frac{n+1}{n}}}.
$$
Recall that $n\geq 1$.
\end{enumerate}

\bigskip

\begin{center}
{\bf (e)} Estimates for $\theta_0$
\end{center}

Obviously, $\hat{\theta}$ is determined over the event horizon by
\begin{equation}
\label{t_eh}
\hat{\theta}(r)=\hat{\frac{\theta}{\lambda}}(r)\hat{\lambda}(r).
\end{equation}
Again, using the bounds in~(c) for $r(0,v)$, we may
bound $\theta_0(v)$ as follows.
\begin{enumerate}[(i)]
\item
In this case, we have
$$
c_1C_1\ln\left({\textstyle\frac{1}{r_+-r_2}}\right)e^{C_1v}
\left({\textstyle\frac{1}{r_+-r_2}}\right)^{-\,\frac{e^{C_2v}}{2}}\leq
|\theta_0|(v)\leq c_2C_2\ln\left({\textstyle\frac{1}{r_+-r_2}}\right)e^{C_2v}
\left({\textstyle\frac{1}{r_+-r_2}}\right)^{-\,\frac{e^{C_1v}}{2}}.
$$
Recall that $r_+-r_2<1$.
\item
Here, we obtain
\begin{equation}\label{theta_ii}
c_1C_1(r_+-r_2)^{\frac{1}{2}}e^{-\frac{C_2}{2}v}\leq|\theta_0|(v)\leq c_2C_2(r_+-r_2)^{\frac{1}{2}}e^{-\frac{C_1}{2}v}.
\end{equation}
\item
In this case, we have
\begin{equation}\label{theta_iii}
\frac{c_1C_1}{\left[\alpha_1C_2v+
\frac{1}{(r_+-r_2)^{\alpha_1}}
\right]^{\frac{1+\alpha_2+(\alpha_2-\alpha_1)}{2\alpha_1}}}
\leq
|\theta_0|(v)\leq \frac{c_2C_2}{\left[\alpha_2C_1v+
\frac{1}{(r_+-r_2)^{\alpha_2}}
\right]^{\frac{1+\alpha_1-(\alpha_2-\alpha_1)}{2\alpha_2}}}.
\end{equation}
Recall that $0<\alpha_1\leq\alpha_2<1$.
\item
Here, we obtain
$$
\frac{c_1C_1}{\left[nC_2v+
\frac{1}{(r_+-r_2)^{n}}\right]^{\frac{n+1}{2n}}}
\leq
|\theta_0|(v)\leq \frac{c_2C_2}{\left[nC_1v+
\frac{1}{(r_+-r_2)^{n}}\right]^{\frac{n+1}{2n}}}.
$$
Recall that $n\geq 1$.
\end{enumerate}

\vspace{6mm}

{\em Recap of the initial conditions}.
To finish this section on the initial conditions, let us summarize the procedure for constructing the initial data.
We start by prescribing the integrable function $\hat f$, which determines
$\hat\varpi$ by~\eqref{omega_eh}.
The $v$ coordinate is fixed by setting $\hat\kappa \equiv 1$,
which in turn yields $\hat\lambda$ by~\eqref{lambda_eh}.
The function $\hat\theta$ is obtained (up to a choice of sign) from~\eqref{tl_eh},
and $r$ is determined at the event horizon by~\eqref{r-v}.
To complete the definition of the initial data, we choose
\begin{equation}
\left\{
\begin{array}{lclcl}
 r(u,0)&=&r_2-u,&&\\
 \nu(u,0)&=&\nu_0(u)&\equiv&-1,\\
 \zeta(u,0)&=&\zeta_0(u),&&
\end{array}
\right.\qquad{\rm for}\ u\in[0,U],
\end{equation}
where $\zeta_0$ is a free continuous function.

\begin{Prop} Under the previous conditions,
if\/ $\hat{f}$ is chosen according to\/ {\rm Hypothesis~\ref{hyp}},
then $\hat\lambda$ is positive in $\left]r_2,r_+\right[$ and
$r(0,\,\cdot\,)$,
$\lambda_0(\,\cdot\,)=\lambda(0,\,\cdot\,)$ and
$\theta_0(\,\cdot\,)=\theta(0,\,\cdot\,)$
have the decay given in\/ {\rm\bf (c)}, {\rm\bf (d)} and\/ {\rm\bf (e)}, respectively.
\end{Prop}

\section{The apparent horizon}\label{section3.5}

Unlike what happens in the Reissner-N\"ordstrom solutions, or in the more general solutions studied in \cite{relIst1, relIst2, relIst3, Dafermos1}, the interiors of the black holes that we are now considering do not coincide with the trapped region, that is, the set of points where $\lambda < 0$. In fact, as a consequence of $\lambda>0$ on the event horizon, there will also exist a regular region ${\cal P}_\lambda$, where $\lambda \geq 0$. It will be shown in this section that the apparent horizon ${\cal A}$, that is, the set of points where $\lambda = 0$, is, in our domain, a $C^1$ curve that separates the regular from the trapped regions. Moreover, we will prove that this curve can be parametrized by $v\mapsto(u_\lambda(v),v)$, with $u_\lambda'\leq 0$ and $\frac{d}{dv}r(u_\lambda(v),v)\geq 0$. We will finish by sketching ${\cal A}$ and the curves where $r$ is constant.

From Theorem 4.4 in \cite{relIst1} we have 
\begin{Thm}\label{maximal} The characteristic initial value problem\/~\eqref{r_u}$-$\eqref{kappa_at_u}, with the initial 
conditions of the previous section
has a unique solution defined on a maximal past set $\mysigma$ containing $[0,\myumax]\times\{0\}\cup\{0\}\times[0,\infty[$.\end{Thm}

We denote by $\Gamma_{\ckr}$ the set where $r$ is equal to $\ckr$. Since $\nu<0$,
this set is a curve that can be parametrized by
$v\mapsto(\uckr(v),v)$. Obviously, 
\begin{equation}\label{r=r}
r(\uckr(v),v)=\ckr.
\end{equation}
Moreover, $r$ is $C^1$, so $\uckr$ is $C^1$. Differentiating both sides of~\eqref{r=r} with respect to $v$ yields
\begin{equation}\label{sobe}
\uckr'(v)=-\,\frac{\lambda(\uckr(v),v)}{\nu(\uckr(v),v)}.
\end{equation}

In Lemma~\ref{sign} we looked at the behavior of $\lambda$ along a line where $u$ is constant.
Now we look at the behavior of $\lambda$ along a curve $\Gamma_{\ckr}$.
\begin{Lem}\label{domain} Fix $\ckr\in\,]0,r_+[$.
If $\lambda(\uckr(\bar v),\bar v)=0$, then $\lambda(\uckr(v),v)\leq 0$ for all $v>\bar v$. In particular, $\uckr(\,\cdot\,)$ is defined in $\left[\bar v,\infty\right[$.
\end{Lem}
\begin{proof}
Let $\bar v$ be such that $\lambda(\uckr(\bar v),\bar v)=0$.
Suppose there exists $v>\bar v$ such that $\lambda(\uckr(v),v)>0$.
Define $\underline{v}=\inf\{\hat{v}:
\lambda(\uckr(\tilde{v}),\tilde{v})>0\ {\rm for\ all}\ \tilde{v}\in\,
]\hat{v},v]\}$. Clearly, $\underline{v}\geq\bar{v}$ and
$\lambda(\uckr(\underline{v}),\underline{v})=0$.
Equality~\eqref{sobe} shows that $\uckr'(\tilde v)>0$ for all
$\tilde{v}\in\,]\underline{v},v]$. 
Thus, $\uckr(\underline{v})<\uckr(v)$.
According to Lemma~\ref{sign},
$\lambda(\uckr(\underline{v}),\tilde v)\leq 0$ for all 
$\tilde{v}\in\,]\underline{v},v]$. This implies
$\ckr\geq r(\uckr(\underline{v}),v)$. Since
$\uckr(\underline{v})<\uckr(v)$ and $\nu<0$, we have
$r(\uckr(\underline{v}),v)>r(\uckr(v),v)$.
However, $r(\uckr(v),v)=\ckr$.
Hence we reached the contradiction $\ckr>\ckr$.

Since $\uckr'(v) \leq 0$ for $v \geq \bar v$, we have $\uckr(v)\leq \uckr(\bar v) < U$. Therefore, only two possibilities could occur to prevent $\uckr(\,\cdot\,)$ from being defined in $\left[\bar v,\infty\right[$: either the curve $\Gamma_{\ckr}$ reaches the event horizon, or the boundary of $\mysigma$ in $[0,\myumax]\times[0,\infty[$. However, the first possibility is excluded because $\lambda$ is strictly positive over the event horizon, and the second by the fact that $r$ goes to zero on the boundary of $\mysigma$ (see \cite{relIst1}).
\end{proof}
Consider 
$$
\bar v=\inf\{v\geq 0:\lambda(\uckr(v),v)\leq 0\}.
$$

\begin{Cor}
One of the following occurs:
\begin{enumerate}[{\rm (a)}]
\item $\bar v=\infty$ and so $\lambda(\uckr(v),v)$ is positive for all $v$.
\item $\bar v\in\R^+_0$ and so $\lambda(\uckr(v),v)$ is positive for 
$v<\bar v$ 
and is nonpositive for $v\geq\bar v$.
\end{enumerate} 
\end{Cor}

From~\eqref{sobe}, in case (b) we have 
$$
\max\left\{\uckr(v):v\in\R^+_0\right\}=\uckr(\bar v).
$$

Recall that the initial data is prescribed so that $\lambda$ is strictly positive on $\{0\}\times[0,\infty[$.
We now choose $U$ sufficiently small so that $\lambda$ is positive on $[0,U]\times\{0\}$.
Let us define the set 
$$
{\cal P}_\lambda:=\{(u,v)\in[0,U]\times[0,\infty[\,:\lambda(u,v)\geq 0\}.
$$
By Lemma~\ref{sign}, if $(u,v) \in {\cal P}_\lambda$ then $\{u\}\times[0,v]\subset{\cal P}_\lambda$. This and the fact that $\nu(u,0)<0$ imply that $$r\geq r(U,0)\ \ {\rm on}\ {\cal P}_\lambda.$$

As noted at the beginning of Section~\ref{section3}, we have $\kappa>0$ and $\nu<0$ over the whole solution domain $\mathcal{P}$. Therefore, $1-\mu=\lambda/\kappa$ is positive on $[0,U]\times\{0\}$, and, by~\eqref{omega_u}, $\min\{\varpi(u,0):0\leq u\leq U\}=\varpi(U,0)$.
Since $\partial_v\varpi\geq 0$, $\varpi$ achieves its minimum at $(U,0)$:
$$
\varpi(U,0)\leq\varpi\ {\rm on}\ {\cal P}.
$$

Since $r(0,0)=r_2>\hat r_0$ (see~\eqref{r0_hat}), we have
$$
\min_{r\in [r_2,r_+]}{\partial_r(1-\mu)}(r,\hat\varpi(r_2))>0.
$$

By further reducing $U$, if necessary, we can make $r(U,0)$ sufficiently close to $r_2$, and therefore
$\varpi(U,0)$ sufficiently close to $\hat{\varpi}(r_2)$, so that
$$
\min_{r\in [r(U,0),r_+]}{\partial_r(1-\mu)}(r,\varpi(U,0))>0.
$$
Since $\partial_r(1-\mu)$ increases with $\varpi$ and $r<r_+$ on $\mathcal{P}$, we then have
$$
\min_{(u,v)\in{\cal P}_\lambda}\partial_r(1-\mu)(u,v)\geq\min_{r\in [r(U,0),r_+]}{\partial_r(1-\mu)}(r,\varpi(U,0))>0.
$$
In addition,
$$
\partial_r(1-\mu)(r(U,0),\varpi_0)>
\partial_r(1-\mu)(r(U,0),\varpi(U,0))>0,
$$ 
and so $r(U,0)>r_0$.

From equation~\eqref{lambda_u}, we have
\begin{equation}\label{dlmz}
\partial_u\lambda<0\ {\rm on}\ {\cal P}_\lambda.
\end{equation} 
Moreover, Lemma~\ref{sign} implies that if $(u,v)$ is such that $\lambda(u,v)=0$, then
$\partial_v\lambda(u,v)\leq 0$. From Lemma~6.1 in~\cite{relIst1} we know that $\lambda$ is $C^1$,
because $\nu_0$, $\kappa_0(\,\cdot\,)=\kappa(0,\,\cdot\,)$ and $\lambda_0$ are $C^1$. We therefore conclude that
if the set 
$$
{\cal A}:=\{(u,v)\in{\cal P}_\lambda:\lambda(u,v)=0\}
$$
is nonempty (as will be seen to be the case in the next section), then it is a $C^1$ manifold
which we can parametrize by $v\mapsto(u_\lambda(v),v)$ with
\begin{equation}\label{ulambdalinha}
u_\lambda'(v)=-\,
\frac{\partial_v\lambda(u_\lambda(v),v)}{\partial_u\lambda(u_\lambda(v),v)}\leq 0.
\end{equation}
To examine the behavior of $r$ along the curve where $\lambda=0$ we compute
\begin{eqnarray}
\frac{d}{dv}r(u_\lambda(v),v)&=&\nu(u_\lambda(v),v)u_\lambda'(v)+\lambda(u_\lambda(v),v)\nonumber\\
&=&\nu(u_\lambda(v),v)u_\lambda'(v)\geq 0.\label{ul}
\end{eqnarray}

Arguing as in the second half of the proof of Lemma~\ref{domain}, we conclude
that the domain of $u_\lambda$ is an interval of the form $[v_0,\infty[$, for some $v_0>0$.

Typically, the (thin) curves of
constant $r$ and the (thick) curve where $\lambda=0$
behave as in the following figure.

\begin{center}
\begin{turn}{45}
\begin{psfrags}
\psfrag{v}{\tiny $v$}
\includegraphics[scale=.7]{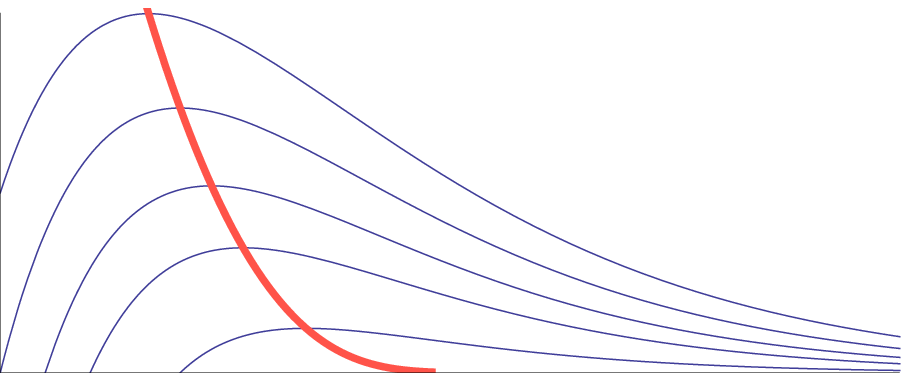}
\end{psfrags}
\end{turn}
\end{center}

Clearly, if a portion of a curve $\Gamma_r$ is parametrized by $v\mapsto(\tilde c,v)$ for $v\in I$, then $\lambda=0$ over that portion of curve. 
Conversely, equation~\eqref{ul} shows that if $r$ is constant over an interval $I$ along the curve
where $\lambda$ is zero, then the portion of this curve over $I$ is parametrized by $v\mapsto(\tilde c,v)$.

Hence, it is not excluded that a curve of constant $r$ and the curve where $\lambda=0$
could partially overlap (as is the case on the event horizon of the Reissner-Nordstr\"{o}m solution).
This is illustrated in the next figure, where again the thin line represents a curve of
constant $r$ and the thick line represents the curve where $\lambda=0$.

\begin{center}
\begin{turn}{45}
\begin{psfrags}
\psfrag{v}{\tiny $v$}
\includegraphics[scale=.7]{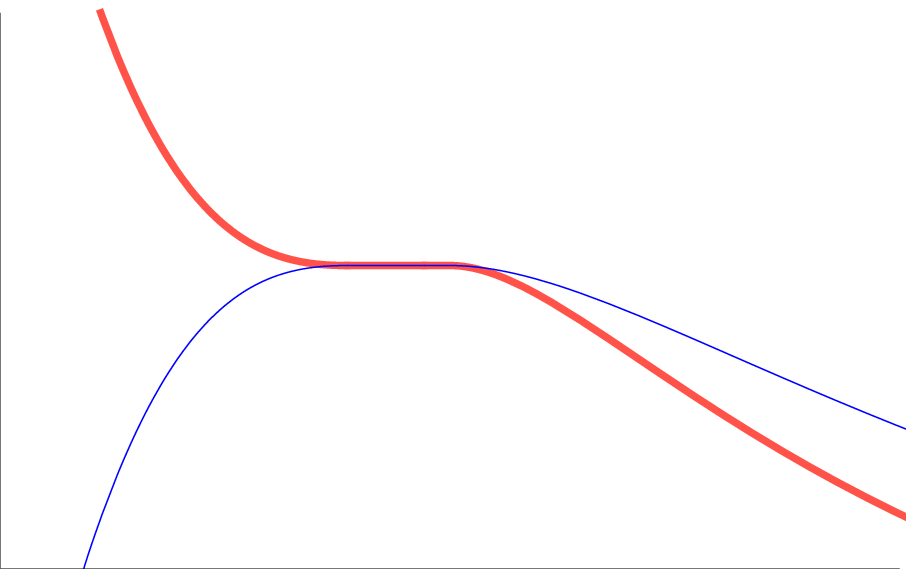}
\end{psfrags}
\end{turn}
\end{center} \label{pagina}

Equation~\eqref{ulambdalinha} implies that ${\cal P}_\lambda$ is a past set. Therefore,
since $\partial_u\varpi\leq 0$ in ${\cal P}_\lambda$, the supremum of $\varpi$ 
in ${\cal P}_\lambda$ is $\varpi_0$:
$$
\varpi\leq\varpi_0\ \ {\rm on}\ {\cal P}_\lambda.
$$

\section{Behavior of the solution on ${\cal P}_\lambda$} \label{section4}

In Hypothesis~\ref{hyp} we listed the main reasonable asymptotics for the renormalized Hawking mass along the event horizon. As was mentioned then, we will now focus exclusively on case~(ii), since it corresponds to the exponential Price law expected for a positive cosmological constant. Note from~\eqref{theta_ii} that in this case
\[
\int_0^\infty |\theta_0|(v)\,dv < \infty.
\]

In this section we will analyze in detail the behavior of the solution in the region ${\cal P}_\lambda$ between the event and the apparent horizons, where $\lambda$ is still nonnegative (in spite of being located inside the black hole), and so the monotonicity properties of the radius and mass functions are the same as outside the black hole. This region does not exist in the Reissner-Nordstr\"om solution, or in the solutions studied in \cite{relIst2,relIst3}, where the two horizons coincide. Therefore we must study the propagation of the decays of the main quantities from the event horizon to the apparent horizon, after which an analysis similar to what was done in \cite{relIst2,relIst3} can be performed. The most significant phenomenon influencing this propagation is the redshift effect.

The estimates in this section depend on accurately controlling the growth of $\frac{\zeta}{\nu}$ in ${\cal P}_\lambda$. The importance of this function stems from the fact that it is a geometric quantity, and so the redshift effect is reflected in its evolution equation \eqref{zeta_nu}. The resulting exponential decay plays a fundamental role in estimating the remaining key quantities.

More precisely, we start by deriving an appropriate two-dimensional version of Gronwall's inequality 
to bound $\frac{\zeta}{\nu}$ in ${\cal P}_\lambda$. We then 
show that the derivative $\partial_r(1-\mu)$ is close to
$\partial_r(1-\mu)(r_+,\varpi_0)= 2k_+$, and $\kappa$ is close to $1$, provided that $u$ is
sufficiently small and $v$ is sufficiently large. This allows us to go back to our
previous estimate for $\frac{\zeta}{\nu}$ and improve it to an exponential decay, dominated by the slower of two competing effects: the redshift arising from the evolution equation, essentially $e^{-2k_+ v}$, and the exponential decay of $\theta_0$.
We can then use it to control the remaining key quantities $\theta$, $\nu$, $u_\lambda$ and $v_\lambda$ (given by~\eqref{vlambda}).
Finally, we show that in ${\cal P}_\lambda$ the radius function $r$ and the renormalized Hawking mass $\varpi$ converge to $r_+$ and
$\varpi_0$, respectively, as $v\to\infty$.

We will start by writing an integral formula for $\frac{\zeta}{\nu}$ in ${\cal P}_\lambda$, which features a crucial dependence on $\theta_0$ and $\frac{\zeta}{\nu}$. Integrating~\eqref{theta_u}, we get
\begin{equation}
\label{bar}
\theta(u,v)=\theta_0(v)-\int_0^u\left[\frac{\zeta}{\nu}\frac{\lambda}{r}\nu\right](\tilde u,v)\,d\tilde u,
\end{equation}
while integrating
\begin{equation}\label{zeta_nu}
\partial_v\left(\frac{\zeta}{\nu}\right)=-\,\frac{\theta}{r}-\kappa\partial_r(1-\mu)\frac{\zeta}{\nu},
\end{equation}
we get
\begin{eqnarray}
\frac{\zeta}{\nu}(u,v)&=&\frac{\zeta}{\nu}(u,0)e^{-\int_0^v[\kappa\partial_r(1-\mu)](u,\tilde v)\,d\tilde v}\nonumber\\
&&-\int_0^v\frac{\theta}{r}(u,\tilde v)e^{-\int_{\tilde v}^v[\kappa\partial_r(1-\mu)](u,\bar v)\,d\bar v}\,d\tilde v.
\label{rafaeli}
\end{eqnarray}
The desired formula for $\frac{\zeta}{\nu}$ is obtained combining~\eqref{bar} with~\eqref{rafaeli}:
\begin{eqnarray}
\frac{\zeta}{\nu}(u,v)&=&\frac{\zeta}{\nu}(u,0)e^{-\int_0^v[\kappa\partial_r(1-\mu)](u,\tilde v)\,d\tilde v}\nonumber\\
&&-\int_0^v\frac{\theta_0(\tilde v)}{r(u,\tilde v)}e^{-\int_{\tilde v}^v[\kappa\partial_r(1-\mu)](u,\bar v)
\,d\bar v}\,d\tilde v\nonumber\\
&&+\int_0^v\int_0^u\left[\frac{\zeta}{\nu}\frac{\lambda}{r}\nu\right](\tilde u,\tilde v)\frac{e^{-\int_{\tilde v}^v[\kappa\partial_r(1-\mu)](u,\bar v)
\,d\bar v}}{r(u,\tilde v)}\,d\tilde ud\tilde v.\label{adele}
\end{eqnarray}

We will need the following version of Gronwall's inequality.
\begin{Lem}\label{our_Gronwall3}
Let $M$ be a positive number, and
assume that $f:\,]0,M]\to[0,\infty[$ is continuous and strictly decreasing with $\lim_{x\searrow 0}f(x)=\infty$.
Consider the set $${\cal S}=\{(x,y)\in\,]0,M]\times[0,\infty[\,:y\geq f(x)\},$$
and continuous functions $c:\,]0,M]\to\,[0,\infty[$ and $u,b:{\cal S}\to[0,\infty[$ such that
\begin{equation}\label{eps}
u(x,y)\leq \eps+\int_{f^{-1}(y)}^xc(\tilde x)\,d\tilde x+\int_{f^{-1}(y)}^x\int_{f(\tilde x)}^yu(\tilde x,\tilde y)b(\tilde x,\tilde y)\,d\tilde yd\tilde x,
\end{equation}
for some positive number $\eps$.
Then
\begin{equation}\label{gronwall3}
u(x,y)\leq \left(\eps+\int_{f^{-1}(y)}^xc(\tilde x)\,d\tilde x\right)e^{\int_{f^{-1}(y)}^x\int_{f(\tilde x)}^yb(\tilde x,\tilde y)\,d\tilde yd\tilde x}.
\end{equation}
\end{Lem}
\begin{proof}
For $(x,y)\in{\cal S}$ we define
$$
v(x,y)=\eps+\int_{f^{-1}(y)}^xc(\tilde x)\,d\tilde x+\int_{f^{-1}(y)}^x\int_{f(\tilde x)}^yu(\tilde x,\tilde y)b(\tilde x,\tilde y)\,d\tilde yd\tilde x
$$
(see the figure below).
\begin{center}
\begin{psfrags}
\psfrag{x}{\!\tiny $x$}
\psfrag{s}{\tiny $\underline{x}$}
\psfrag{f}{\tiny $\!\!\!\!\!\!\!f(\underline{x})$}
\psfrag{g}{\tiny $\!\!\!\! y$}
\psfrag{a}{\tiny $\!\!\!\!\!\!\!\!\!\!{f^{-1}(y)}$}
\psfrag{M}{\tiny $\!\! M$}
\includegraphics[scale=.5]{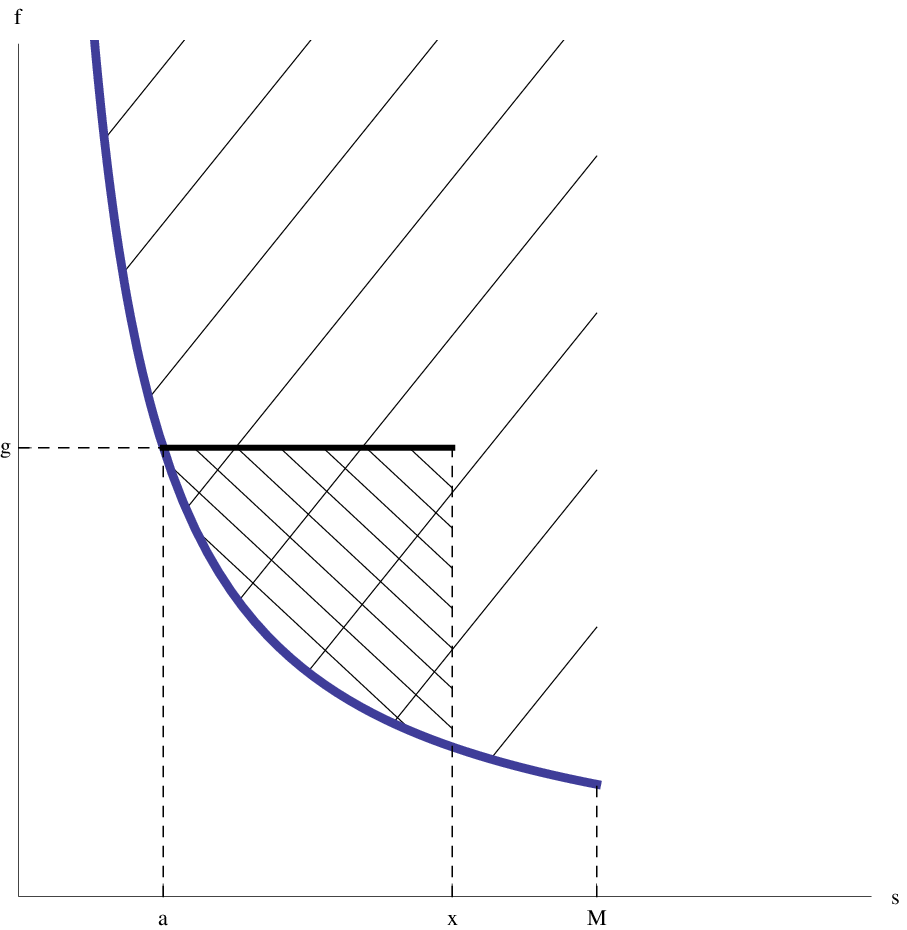}
\end{psfrags}
\end{center}
According to our hypothesis, $u\leq v$. Note that if $\tilde y\leq y$ then $v(x,\tilde y)\leq v(x,y)$, because when we 
change $\tilde y$ to $y$ we are integrating nonnegative
functions over larger domains. Hence, as $b$ is nonnegative,
\begin{eqnarray*}
\partial_x v(x,y)&=&c(x)+\int_{f(x)}^yu(x,\tilde y)b(x,\tilde y)\,d\tilde y\\
&\leq&c(x)+\int_{f(x)}^yv(x,\tilde y)b(x,\tilde y)\,d\tilde y\\
&\leq&c(x)+v(x,y)\int_{f(x)}^yb(x,\tilde y)\,d\tilde y.
\end{eqnarray*}
Next we use $v(x,y)\geq\eps+\int_{f^{-1}(y)}^xc(\tilde x)\,d\tilde x>0$. We may write
$$
\frac{\partial_x v(\tilde x,y)}{v(\tilde x,y)}\leq\frac{c(\tilde x)}{\eps+\int_{f^{-1}(y)}^{\tilde x}c(\bar x)\,d\bar x}+\int_{f(\tilde x)}^yb(\tilde x,\tilde y)\,d\tilde y.
$$
Integrating both sides of the last inequality in $\tilde x$, from $f^{-1}(y)$ to $x$, we get
\begin{eqnarray*}
\ln v(x,y)&\leq&\ln v(f^{-1}(y),y)\\ &&+\ln\left(\eps+\int_{f^{-1}(y)}^xc(\tilde x)\,d\tilde x\right)-
\ln\eps\\
&&+\int_{f^{-1}(y)}^x\int_{f(\tilde x)}^yb(\tilde x,\tilde y)\,d\tilde yd\tilde x.
\end{eqnarray*}
Taking into account that
$
v(f^{-1}(y),y)=\eps,
$
$$
v(x,y)\leq \left(\eps+\int_{f^{-1}(y)}^xc(\tilde x)\,d\tilde x\right)e^{\int_{f^{-1}(y)}^x\int_{f(\tilde x)}^yb(\tilde x,\tilde y)\,d\tilde yd\tilde x}.
$$
Since $u(x,y)\leq v(x,y)$, we obtain~\eqref{gronwall3}.
\end{proof}
\begin{Cor}\label{Corollary-0}
Under the hypotheses of\/ {\rm Lemma~\ref{our_Gronwall3}}, if
\begin{equation}\label{zero-0}
u(x,y)\leq \int_{f^{-1}(y)}^xc(\tilde x)\,d\tilde x+\int_{f^{-1}(y)}^x\int_{f(\tilde x)}^yu(\tilde x,\tilde y)b(\tilde x,\tilde y)\,d\tilde yd\tilde x,
\end{equation}
then
\begin{equation}\label{gronwall4}
u(x,y)\leq \left(\int_{f^{-1}(y)}^xc(\tilde x)\,d\tilde x\right)e^{\int_{f^{-1}(y)}^x\int_{f(\tilde x)}^yb(\tilde x,\tilde y)\,d\tilde yd\tilde x}.
\end{equation}
\end{Cor}
\begin{proof}
If $u$ satisfies~\eqref{zero-0}, then $u$ satisfies~\eqref{eps}, for every positive $\eps$. The result is obtained by letting $\eps\searrow 0$.
\end{proof}

A small variation of Lemma~\ref{our_Gronwall3} is
\begin{Lem}\label{our_Gronwall4}
Let $M$ be a positive number
and consider continuous functions $c:[0,M]\to\,[0,\infty[$ and $u,b:[0,M]\times[0,\infty[\,\to[0,\infty[$ such that
\begin{equation}\label{epsbis}
u(x,y)\leq \eps+\int_{0}^xc(\tilde x)\,d\tilde x+\int_{0}^x\int_{0}^yu(\tilde x,\tilde y)b(\tilde x,\tilde y)\,d\tilde yd\tilde x,
\end{equation}
for some positive number $\eps$.
Then
\begin{equation}\label{gronwall3bis}
u(x,y)\leq \left(\eps+\int_{0}^xc(\tilde x)\,d\tilde x\right)e^{\int_{0}^x\int_{0}^yb(\tilde x,\tilde y)\,d\tilde yd\tilde x}.
\end{equation}
\end{Lem}

We will now apply Lemma~\ref{our_Gronwall4} to~\eqref{adele} in order to bound $\frac{\zeta}{\nu}$ in ${\cal P}_\lambda$. We have seen that on ${\cal P}_\lambda$ we have $r\geq r(U,0)>r_0$,
$\varpi\leq\varpi_0$ and $\partial_r(1-\mu)\geq c>0$. Since the exponentials in~\eqref{adele} are bounded above by $1$ and $\frac{1}{r}$ is bounded
above by $\frac{1}{r(U,0)}$, we get
\begin{eqnarray}
\Bigl|\frac{\zeta}{\nu}\Bigr|(u,v)
&\leq& C\left(\sup_{[0,U]}|\zeta_0|+\frac{1}{r(U,0)}\int_0^\infty|\theta_0|(\tilde v)\,d\tilde v\right)=:C_\lambda.
\label{constant0}
\end{eqnarray}
Indeed, as $\lambda$ is nonnegative and $\partial_u\lambda<0$ on ${\cal P}_\lambda$,
\begin{eqnarray*}
&&\int_0^v\int_0^u\left[\frac{\lambda}{r}(-\nu)\right](\tilde u,\tilde v)\,d\tilde ud\tilde v\leq
\int_0^v\int_0^u\left[\frac{-\nu}{r}\right](\tilde u,\tilde v)\lambda(0,\tilde v)\,d\tilde ud\tilde v\\
&&\ \ \ =
\int_0^v\lambda(0,\tilde v)\ln\left(\frac{r(0,\tilde v)}{r(u,\tilde v)}\right)\,d\tilde v
\leq
\int_{r(0,0)}^{r(0,v)}\ln\left(\frac{r(0,\tilde v)}{r(U,0)}\right)\,dr(0,\tilde v)\\
&&\ \ \  =r(0,v)\ln\left(\frac{r(0,v)}{r(U,0)}\right)-r(0,v)-r(0,0)\ln\left(\frac{r(0,0)}{r(U,0)}\right)+r(0,0)\\
&&\ \ \ <r_+\ln\left(\frac{r_+}{r(U,0)}\right)-r_+-r(0,0)\ln\left(\frac{r(0,0)}{r(U,0)}\right)+r(0,0)\\
&&\ \ \ =\ln\left(\frac{c(r(U,0),r_+)}{r(U,0)}\right)(r_+-r(0,0))\\
&&\ \ \ < \ln\left(\frac{r_+}{r(U,0)}\right)(r_+-r(0,0)).
\end{eqnarray*}
Here $c(r(U,0),r_+)\in\,]r(U,0),r_+[$ is provided by the Mean Value Theorem.
The constant $C$ in~\eqref{constant0} is bounded by
$$
C\leq\left(\frac{r_+}{r(U,0)}\right)^{\frac{r_+-r(0,0)}{r(U,0)}}.
$$

Arguing as in Section~\ref{section3.5}, one can easily see that $r(U,V)\leq r\leq r_+$ and $\varpi(U,V)\leq\varpi\leq\varpi_0$ in the set 
$$
\ruv:=\{(u,v)\in {\cal P}_\lambda:0\leq u\leq U\ {\rm and}\ v\geq V\}.
$$
Since $\lim_{v\to\infty}r(0,v)=r_+$ and $\lim_{v\to\infty}\varpi(0,v)=\varpi_0$, the continuity of the functions $r$ and $\varpi$ guarantees that they are close to $r_+$ and $\varpi_0$ in $\ruv$ if we choose $U$ sufficiently small and $V$ sufficiently large. Therefore, given $\delta>0$, there exist $U>0$ and $V\geq 0$ such that
$$
2\kk-\delta\leq
\partial_r(1-\mu)\leq 2\kk+\delta
$$
in $\ruv$.
Integrating~\eqref{kappa_u} in $\ruv$, we obtain
\begin{equation}\label{c_k}
C_\kappa:=\left(\frac{r(U,V)}{r_+}\right)^{C_\lambda^2}\leq\kappa\leq 1.
\end{equation}
Notice that $C_\kappa$ can be made arbitrarily close to one by choosing $U$ small and $V$ large.
So, in this set,
\begin{align}
-C_\alpha:=-(2\kk +\delta)&=-(\partial_r(1-\mu)(r_+,\varpi_0)+\delta)\nonumber\\
&\leq-\kappa\partial_r(1-\mu)\leq\label{c-alpha}\\
&-C_\kappa(\partial_r(1-\mu)(r_+,\varpi_0)-\delta)=-C_\kappa(2\kk -\delta)=:-c_\alpha.\nonumber
\end{align}

We will now improve our estimate for $\frac{\zeta}{\nu}$ to an exponential decay. Going back to~\eqref{adele}, the first exponential is bounded above by
$$
e^{-\int_0^v[\kappa\partial_r(1-\mu)](u,\tilde v)\,d\tilde v}\leq
e^{-c_\alpha v},
$$
while the second and third exponentials are bounded above by
$$
e^{-\int_{\tilde v}^v[\kappa\partial_r(1-\mu)](u,\bar v)\,d\bar v}\leq
e^{-c_\alpha(v-\tilde v)}.
$$

Applying again Lemma~\ref{our_Gronwall4}, this time 
to $e^{c_\alpha v}\frac{\zeta}{\nu}(u,v)$
(see also the proof of Lemma~4.1 of \cite{relIst2}),
we obtain, in $\ruv$,
\begin{eqnarray}
\Bigl|\frac{\zeta}{\nu}\Bigr|(u,v)
&\leq& C\left(\sup_{[0,U]}\Bigl|\frac\zeta\nu\Bigr|(\,\cdot\,,V)+\int_V^v e^{c_\alpha\tilde v}|\theta_0|(\tilde v)\,d\tilde v\right)e^{-c_\alpha v}.
\label{constant10}
\end{eqnarray}
In view of the decay~\eqref{theta_ii} for $\theta_0$, $V$ can be chosen so that the indefinite integral 
$\int_V^\infty e^{c_\alpha\tilde v}|\theta_0|(\tilde v)\,d\tilde v$
converges provided $c_\alpha<\frac{C_1}{2}$. 
We define
\begin{equation}\label{s}
\s:=\frac{A}{r_+\kk} - 1 > 0,
\end{equation}
the normalized distance from $A$ to its minimum allowed value (see \eqref{A}). Note that in case~(ii) this distance must be positive. The value of $C_1$ is expressed in terms of $s$ by (see~\eqref{C1})
$$
\textstyle
C_1=2\s\kk-\delta,
$$
and so $c_\alpha<\frac{C_1}{2}$ amounts to
\begin{equation}\label{AA}
s>2.
\end{equation}
In this case~\eqref{constant10} yields
\begin{equation}\label{constant}
\Bigl|\frac{\zeta}{\nu}\Bigr|(u,v)\leq
C\left(\sup_{[0,U]}\Bigl|\frac\zeta\nu\Bigr|(\,\cdot\,,V)+1\right)e^{-c_\alpha v}
\end{equation}
When $\s\leq 2$, we obtain from~\eqref{constant10}
\begin{equation}\label{three_1}
\Bigl|\frac{\zeta}{\nu}\Bigr|(u,v)\leq Ce^{-(\s\kk-\delta)v},
\end{equation}
where we still have exponential decay. The existence of these two different regimes reflects the competition phenomenon, mentioned at the beginning of this section, between the redshift arising from the evolution equation (which dominates for $\s>2$) and the exponential decay of $\theta_0$ (dominant for $\s<2$). \label{regime}

We now use this improved estimate for $\bigl|\frac{\zeta}{\nu}\bigr|$ to control the remaining quantities, starting with $\theta$. Using~\eqref{dlmz} and~\eqref{bar}, we have, for $(u,v)\in{\cal P}_\lambda$,
\begin{eqnarray*}
|\theta(u,v)-\theta_0(v)|&\leq& \lambda_0(v)\max_{[0,u]\times\{v\}}\Bigl|\frac{\zeta}{\nu}\Bigr|
\int_0^u\left[\frac{-\nu}{r}\right](\tilde u,v)\,d\tilde u\\
&\leq&\ln\left({\textstyle\frac{r_+}{r(U,0)}}\right)
\lambda_0(v)\max_{[0,u]\times\{v\}}\Bigl|\frac{\zeta}{\nu}\Bigr|.
\end{eqnarray*}
This yields, from~\eqref{l_ii} and~\eqref{constant0},
$$
|\theta(u,v)-\theta_0(v)|\leq Ce^{-C_1v}.
$$
In view of~\eqref{theta_ii}, we conclude that, for $(u,v)\in\ruv$,
\begin{equation}\label{theta_l_ii}
ce^{-\,\frac{C_2}{2}v}\leq|\theta|(u,v)\leq Ce^{-\,\frac{C_1}{2}v}.
\end{equation}
Note that the decay of $\theta$ is faster than that of $\bigl|\frac{\zeta}{\nu}\bigr|$ for $\s>2$ due to the exponential decay of $\lambda_0$. This effect is lost when we cross the apparent horizon, as will be seen in the next section.

Using the integrated form of~\eqref{nu_v},
$$
\nu(u,v)=\nu(u,V)e^{\int_V^v[\kappa\partial_r(1-\mu)](u,\tilde v)\,d\tilde v},
$$
we conclude that
\begin{equation}\label{nu_A}
-Ce^{C_\alpha (v-V)}\leq\nu(u,v)\leq-ce^{c_\alpha (v-V)}
\end{equation}
in $\ruv$.
Here $$-C=\min\{\nu(u,V):u\in[0,U]\}\leq\max\{\nu(u,V):u\in[0,U]\}=-c<0.$$

We will now estimate $u_\lambda$. We start by noticing that, using~\eqref{lambda_u},
$$
-C_\alpha Ce^{C_\alpha (v-V)}\leq\partial_u\lambda(u,v)\leq-c_\alpha c e^{c_\alpha (v-V)}
$$
in the set $\ruv$. Moreover, integrating $\partial_u\lambda$ from the event horizon to the apparent horizon, we obtain
\begin{equation}\label{lambda}
0=\lambda_0(v)+\int_0^{u_\lambda(v)}\partial_u\lambda(\tilde u,v)\,d\tilde u.
\end{equation}

Since (see~\eqref{l_ii})
$$
ce^{-C_2v}\leq\lambda_0(v)\leq Ce^{-C_1v},
$$
we deduce that
\begin{equation}\label{u_lambda_ii}
\frac{ce^{C_\alpha V}}{C_\alpha C} e^{-(C_2+C_\alpha)v}\leq u_\lambda(v)\leq
\frac{Ce^{c_\alpha V}}{c_\alpha c} e^{-(C_1+c_\alpha)v}.
\end{equation}
Let $\delta>0$. Our parameters can be chosen so that (see~\eqref{C1}, \eqref{C2} and~\eqref{c-alpha})
$$
\frac{2A}{r_+}-\delta<C_1+c_\alpha< C_2+C_\alpha<\frac{2A}{r_+}+\delta.
$$
In particular, the exponents in~\eqref{u_lambda_ii} are positive, and consequently ${\cal A}$ is nonempty.

For $u\in\,]0,U]$, we define
\begin{equation}\label{vlambda}
v_\lambda(u)=\min\{v:\lambda(u,v)=0\}.
\end{equation}
Using $u_\lambda(v_\lambda(u))=u$ in~\eqref{u_lambda_ii}, for each $\delta>0$ we have
\begin{equation}\label{v_lambda_ii}
\left(\frac{r_+}{2A}-\delta\right)\ln\left(\frac{c}{u}\right)\leq v_\lambda(u)\leq
\left(\frac{r_+}{2A}+\delta\right)\ln\left(\frac{C}{u}\right).
\end{equation}

We now characterize the behavior of $r$ on ${\cal A}$.
Taking into account~\eqref{lambda_u}, \eqref{c-alpha} and~\eqref{lambda}, we have
\begin{equation}\label{good}
-\,\frac{\lambda_0(v)}{c_\alpha}\leq
\int_{0}^{u_\lambda(v)}\nu(\tilde u,v)\,d\tilde u\leq-\,\frac{\lambda_0(v)}{C_\alpha},
\end{equation}
that is,
$$
\frac{\lambda_0(v)}{C_\alpha}\leq
r(0,v)-r(u_\lambda(v),v)\leq\frac{\lambda_0(v)}{c_\alpha}.
$$
This implies that
\begin{equation}\label{r_A}
\lim_{v\to\infty}r(u_\lambda(v),v)=r_+,
\end{equation}
and so also
\begin{equation}
\lim_{\stackrel{(u,v)\in{\cal P}_\lambda}{{\mbox{\tiny{$v\to\infty$}}}}}r(u,v)=r_+.
\end{equation}

Examining the sign of the components of $d\varpi$, we see that for small $v$
the level curves of $\varpi$ are qualitatively like the ones in the following figure, where the thick curve represents ${\cal A}$. Notice that
\begin{eqnarray}
\frac{d}{dv}[\varpi(u_\lambda(v),v)]&=&\partial_u\varpi(u_\lambda(v),v)u_\lambda'(v)
+\partial_v\varpi(u_\lambda(v),v)\nonumber\\
&=&\partial_v\varpi(u_\lambda(v),v)\ \geq\ 0.\label{grow}
\end{eqnarray}

\begin{center}
\begin{turn}{45}
\begin{psfrags}
\psfrag{v}{\tiny $v$}
\includegraphics[scale=.6]{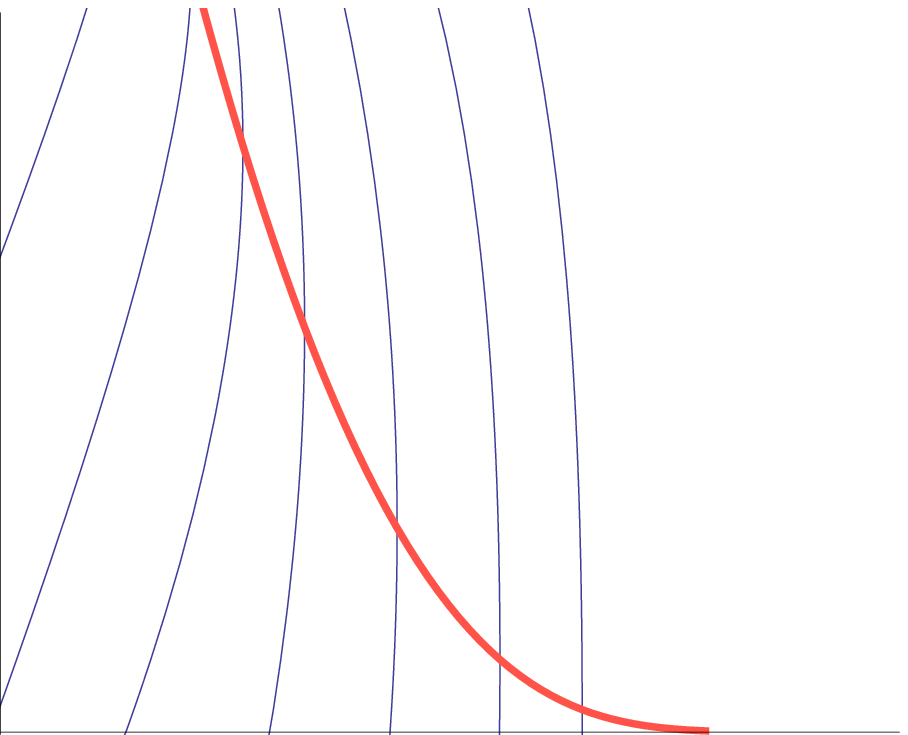}
\end{psfrags}
\end{turn}
\end{center}

The integrated form of~\eqref{omega_u} is
\begin{eqnarray} 
\varpi(u,v)=\varpi(0,v)e^{-\int_0^u\bigl(\frac{\zeta^2}{r\nu}\bigr)(\tilde u,v)\,d\tilde u}\qquad\qquad\qquad\qquad\qquad\qquad\qquad&&\nonumber\\ \qquad\qquad\ \  +
\int_0^ue^{-\int_{\tilde u}^u\frac{\zeta^2}{r\nu}(\bar u,v)\,d\bar u}\left(\frac{1}{2}\left(1+\frac{e^2}{r^2}
-\frac{\Lambda}{3}r^2\right)\frac{\zeta^2}{\nu}\right)(\tilde u,v)\,d\tilde u.&&\label{omega_final}
\end{eqnarray}
The rough estimate (see \eqref{constant0} and \eqref{r_A})
$$
0\leq -\int_0^{u_\lambda(v)}\frac{\zeta^2}{r\nu}(\tilde u,v)\,d\tilde u\leq
C_\lambda^2\ln\left(\frac{r(0,v)}{r(u_\lambda(v),v)}\right) = o(1)
$$
implies that,
for $0\leq u\leq u_\lambda(v)$, the second term on the right-hand side of \eqref{omega_final} is also $o(1)$ as $v \to \infty$, and so
\begin{equation}\label{o1}
\varpi(u,v)=\varpi(0,v)+o(1).
\end{equation}
This yields
\begin{equation}\label{zero}
\lim_{\stackrel{(u,v)\in{\cal P}_\lambda}{{\mbox{\tiny{$v\to\infty$}}}}}\varpi(u,v)=\varpi_0.
\end{equation}

\section{The region $J^-(\Gamma_{\ckrp})\cap J^+({\cal A})$}\label{section5}

From this point on we will consider the solution defined on the intersection of the maximal past set $\cal P$ with the rectangle $\left[0,U\right]\times\left[V,\infty\right[$, for suitably chosen $U>0$ and $V\geq 0$.
In this section we focus on the subset $J^-(\Gamma_{\ckrp})\cap J^+({\cal A})$, for a given ${\ckrp} \in \left]r_0,r_+\right[$, which will later be set conveniently close to $r_+$. We will see that the exponential decays along the apparent horizon of $\frac{\zeta}{\nu}$ and $\theta$ persist in this new region. However, the faster decay rate of $\theta$ in the case $\s>2$ is lost, dominated by that of $\frac{\zeta}{\nu}$. As in Section~4 of~\cite{relIst2}, the solution here still behaves qualitatively as the Reissner-Nordstr\"om solution: $\varpi$ is close to $\varpi_0$, $\kappa$ is close to $1$ and $\zeta, \theta$ are close to $0$. Besides, the approximation improves by making $U$ smaller and $V$ larger.

More precisely, we start by showing that, given $\delta>0$ small, we can choose $V$ sufficiently large
and $U$ sufficiently small so that $\varpi\geq\varpi_0-\delta$ in
$J^+({\cal A})$ and $\partial_r(1-\mu)>0$ in $J^-(\Gamma_{\ckrp})\cap J^+({\cal A})$. This implies that $\partial_u\lambda<0$, and so we can use the two-dimensional version
of Gronwall's inequality in Lemma~\ref{our_Gronwall3} to estimate $\frac{\zeta}{\nu}$.
This allows us to control $\kappa$ from below and $\varpi$ from above,
which, as before, leads to an improved estimate for $\frac{\zeta}{\nu}$.
We then go on to bound $\nu$, $u_{\ckrp}$ and $v_{\ckrp}$.
The bounds on $r$ and $\varpi$ enable us to determine the precise
behavior of $1-\mu$ and, consequently, of $\lambda$. Finally,
we obtain bounds for $\theta$, which are quantitatively like those
for~$\frac{\zeta}{\nu}$.

Let us choose
\begin{equation}\label{rp-delta}
0<\delta<\frac{\ckrp^2}{2}\min_{r\in[\ckrp,r_+]}\partial_r(1-\mu)(r,\varpi_0).
\end{equation}
It is clear from~\eqref{r_A} and~\eqref{zero} that there exists $V$ such that
$$
\left\{
\begin{array}{rcl}
\varpi(u_\lambda(v),v)&\geq&\varpi_0-\delta,\\
r(u_\lambda(v),v)&\geq&\ckrp,
\end{array}
\right.\ \ {\rm for}\ v\geq V.
$$
We choose $0<U\leq u_\lambda(V)$. Then, from \eqref{omega_v},
$$
\varpi\geq\varpi_0-\delta \ \ {\rm in}\ J^+({\cal A}).
$$
It follows that
$$
-\partial_r(1-\mu)(r,\varpi)\leq -\partial_r(1-\mu)(r,\varpi_0)+
\frac{2\delta}{r^2} \ \ {\rm in}\ J^+({\cal A}),
$$ 
and, recalling from \eqref{ul} that $r$ increases along ${\cal A}$, and so $r \leq r_+$ on $J^+({\cal A})$, 
$$
-\partial_r(1-\mu)(r,\varpi)\leq -\min_{r\in[\ckrp,r_+]}\partial_r(1-\mu)(r,\varpi_0)+
\frac{2\delta}{\ckrp^2} \ \ {\rm in}\ J^-(\Gamma_{\ckrp})\cap J^+({\cal A}).
$$ 
Since~\eqref{rp-delta} holds, 
\begin{equation}
\label{dlmz_2}
-\partial_r(1-\mu)(r,\varpi)<0 \ \ {\rm in}\ J^-(\Gamma_{\ckrp})\cap J^+({\cal A}).
\end{equation}

To estimate $\frac{\zeta}{\nu}$ for $(u,v)\in J^+({\cal A})$, we use the expression (similar to~\eqref{adele})
\begin{eqnarray}
\frac{\zeta}{\nu}(u,v)&=&\frac{\zeta}{\nu}(u,v_\lambda(u))e^{-\int_{v_\lambda(u)}^v[\kappa\partial_r(1-\mu)](u,\tilde v)\,d\tilde v}\nonumber\\
&&-\int_{v_\lambda(u)}^v\frac{\theta({u_\lambda(\tilde v)},\tilde v)}{r(u,\tilde v)}e^{-\int_{\tilde v}^v[\kappa\partial_r(1-\mu)](u,\bar v)
\,d\bar v}\,d\tilde v\nonumber\\
&&+\int_{v_\lambda(u)}^v\int_{u_\lambda(\tilde v)}^u\left[\frac{\zeta}{\nu}\frac{\lambda}{r}\nu\right](\tilde u,\tilde v)\frac{e^{-\int_{\tilde v}^v[\kappa\partial_r(1-\mu)](u,\bar v)
\,d\bar v}}{r(u,\tilde v)}\,d\tilde ud\tilde v\label{shakira}
\end{eqnarray}
(recall the definition of $v_\lambda$ in~\eqref{vlambda}).
From \eqref{dlmz_2} we conclude that $\partial_u\lambda<0$ for $(u,v)\in J^-(\cg_{\ckrp})\cap J^+({\cal A})$, and so we have
\begin{eqnarray*}
&&\int_{v_\lambda(u)}^v\int_{u_\lambda(\tilde v)}^u
\left[\frac{(-\nu)}{r}(-\lambda)\right](\tilde u,\tilde v)\,d\tilde ud\tilde v\\
&&\qquad\qquad\qquad\leq\int_{v_\lambda(u)}^v\int_{u_\lambda(\tilde v)}^u
\frac{(-\nu)}{r}(\tilde u,\tilde v)(-\lambda)(u,\tilde v)\,d\tilde ud\tilde v\\
&&\qquad\qquad\qquad\leq\ln\left(\frac{r_+}{\ckrp}\right)\int_{v_\lambda(u)}^v
(-\lambda)(u,\tilde v)\,d\tilde v\\
&&\qquad\qquad\qquad\leq\ln\left(\frac{r_+}{\ckrp}\right)(r(u,v_\lambda(u))-r(u,v))\\
&&\qquad\qquad\qquad\leq\ln\left(\frac{r_+}{\ckrp}\right)(r_+-\ckrp).
\end{eqnarray*}

We apply a generalized version of Lemma~\ref{our_Gronwall3} (because $f(v)=u_\lambda(v)$
might not be strictly decreasing) whose proof we leave to the reader (just approximate
$u_\lambda$ by a strictly decreasing function and pass to the limit).
For $(u,v)\in J^-(\cg_{\ckrp})\cap J^+({\cal A})$,
\begin{eqnarray}
&&\Bigl|\frac{\zeta}{\nu}\Bigr|(u,v)\leq\label{zeta_nu_3}\\
&&\ \
\left(\frac{r_+}{\ckrp}\right)^{\frac{r_+-\ckrp}{\ckrp}}
 \left(\sup_{\left]0,U\right]}\Bigl|\frac{\zeta}{\nu}\Bigr|(u,v_\lambda(u))+\frac{1}{\ckrp}
 \int_{v_\lambda(u)}^\infty|\theta|(u_\lambda(\tilde v),\tilde v)\,d\tilde v\right)=:C_{\ckrp}.
\nonumber
\end{eqnarray}

We see from~\eqref{constant}, \eqref{three_1} and \eqref{theta_l_ii} that $C_{\ckrp}$ is finite.

Integrating~\eqref{kappa_u} yields
$$
\kappa(u,v)=\kappa(u_\lambda(v),v)e^{\int_{u_\lambda(v)}^u\bigl[\left(\frac{\zeta}{\nu}\right)^2
\frac{\nu} r\bigr](\tilde u,v)\,d\tilde u},
$$
and so, using~\eqref{c_k}, we obtain in $J^-(\Gamma_{\ckrp})\cap J^+({\cal A})$
\begin{equation}\label{kappa_2}
C_{\kappa,2}:=C_\kappa\left(\frac{\ckrp}{r_+}\right)^{C_{\ckrp}^2}\leq\kappa\leq 1.
\end{equation}

From~\eqref{zero}, we have 
\begin{eqnarray}
\varpi_m&:=&\inf\{\varpi(u,v):(u,v)\in J^+({\cal A})\ {\rm and}\ v\geq V\}
\label{omega_m}\\ 
&\,=&\varpi(U,v_\lambda(U))\ =\ \varpi_0+o(1),\nonumber
\end{eqnarray}
as $V\to\infty$ (recall that $0<U\leq u_\lambda(V)$). Analogously to~\eqref{omega_final}, we now have
\begin{eqnarray} 
\varpi(u,v)=\varpi(u_\lambda(v),v)e^{-\int_{u_\lambda(v)}^u\bigl(\frac{\zeta^2}{r\nu}\bigr)(\tilde u,v)\,d\tilde u}\qquad\qquad\qquad\qquad\qquad&&\nonumber\\ \qquad\qquad\ \  +
\int_{u_\lambda(v)}^ue^{-\int_{\tilde u}^u\frac{\zeta^2}{r\nu}(\bar u,v)\,d\bar u}\left(\frac{1}{2}\left(1+\frac{e^2}{r^2}
-\frac{\Lambda}{3}r^2\right)\frac{\zeta^2}{\nu}\right)(\tilde u,v)\,d\tilde u.&&\label{omega_final_2}
\end{eqnarray}
Using~\eqref{o1}, \eqref{zeta_nu_3} and \eqref{omega_final_2}, multiplying and dividing by $\nu$ as needed, we have
\begin{eqnarray}
\varpi_M&:=&\sup\{\varpi(u,v):(u,v)\in J^-(\Gamma_{\ckrp})\cap J^+({\cal A})\}
\label{omega_M}\\
&\,\leq&\varpi_0+o(1),\nonumber
\end{eqnarray}
as $\ckrp\nearrow r_+$.
Thus, given $\delta>0$ we can choose $V$ sufficiently large,
$0<U\leq u_\lambda(V)$
and $\ckrp$ sufficiently close to $r_+$
so that, for $(u,v)\in J^-(\Gamma_{\ckrp}) \cap J^+({\cal A})$, we have
\begin{align}
& -C_{\alpha,2}:=-2\kk-\delta=-\partial_r(1-\mu)(r_+,\varpi_0)-\delta\leq\label{C_alpha_2}\\
& \qquad\qquad\ -\max_{r\in[\ckrp,r_+]}\partial_r(1-\mu)(r,\varpi_M)\leq\nonumber\\
& \qquad\qquad\qquad\ \ -\kappa\partial_r(1-\mu)(u,v)\nonumber\\
&\qquad\qquad\qquad\qquad\leq-C_{\kappa,2}\min_{r\in[\ckrp,r_+]}\partial_r(1-\mu)(r,\varpi_m)\nonumber\\
& \qquad\qquad\qquad\qquad\leq-\partial_r(1-\mu)(r_+,\varpi_0)+\delta=-2\kk+\delta=:-c_{\alpha,2}.\label{c_alpha_2}
\end{align}

Applying again a generalized version of Lemma~\ref{our_Gronwall3}, this time 
to $e^{c_{\alpha,2} v}\frac{\zeta}{\nu}(u,v)$
(as was done in~\eqref{constant10}), and carefully taking
the supremum over the exact interval $[u_\lambda(v),u]$ (due to the unboundedness of the exponential term over the apparent horizon), leads to
\begin{eqnarray}
&&\Bigl|\frac{\zeta}{\nu}\Bigr|(u,v)\leq\label{nova}\\
&& C\left(\sup_{\tilde u\in[u_\lambda(v),u]}\Bigl|\frac\zeta\nu\Bigr|(\tilde u,v_\lambda(\tilde u))
e^{c_{\alpha,2}v_\lambda(\tilde u)}+
\int_{v_\lambda(u)}^v e^{c_{\alpha,2}\tilde v}|\theta|(u_\lambda(\tilde v),\tilde v)\,d\tilde v\right)e^{-c_{\alpha,2} v}.
\nonumber
\end{eqnarray}
Suppose first that $c_{\alpha,2}<\frac{C_1}{2}$, which amounts to
$\s>2$ (see~\eqref{AA}).
For the first term in~\eqref{nova} we use~\eqref{constant}
and for the second we use
the decay~\eqref{theta_l_ii} for $\theta$ along the apparent horizon. We get
\begin{eqnarray}
&&\Bigl|\frac{\zeta}{\nu}\Bigr|(u,v)\label{zn_ii_a}
\leq Ce^{-(2\kk-\delta) v}.
\end{eqnarray}
Here and elsewhere we use $\delta$ to mean a parameter that can be made arbitrarily small by taking $U$ sufficiently small, $V$ sufficiently large and $\ckrp$ sufficiently close to $r_+$. In this last inequality, it collects all the previous small quantities, also denoted by $\delta$, arising in the products of the exponentials.

In the case that $\s<2$, we use~\eqref{three_1} and~\eqref{theta_l_ii} to obtain
\begin{eqnarray}
&&\Bigl|\frac{\zeta}{\nu}\Bigr|(u,v)\label{constant4_new}
\leq Ce^{-(\s\kk-\delta)v}.
\end{eqnarray}

From 
$$
\nu(u,v)=\nu(u,v_\lambda(u))e^{\int_{v_\lambda(u)}^v[\kappa\partial_r(1-\mu)](u,\tilde v)\,d\tilde v},
$$
and using \eqref{nu_A}, \eqref{C_alpha_2} and~\eqref{c_alpha_2},
we conclude that
$$
-Ce^{C_\alpha (v_\lambda(u)-V)}e^{C_{\alpha,2} (v-v_\lambda(u))}\leq\nu(u,v)
\leq-ce^{c_\alpha (v_\lambda(u)-V)}e^{c_{\alpha,2} (v-v_\lambda(u))},
$$
implying
\begin{equation}
\label{nu_2}
-Ce^{C_{\alpha,3} v}\leq\nu(u,v)
\leq-ce^{c_{\alpha,3} v},
\end{equation}
where $c_{\alpha,3}=\min\{c_\alpha,c_{\alpha,2}\}$ and $C_{\alpha,3}=\max\{C_\alpha,C_{\alpha,2}\}$.

Using
$$
\ckrp-(r_++o(1))=\ckrp-r(u_\lambda(v),v)=\int_{u_\lambda(v)}^{\ugrp(v)}\nu(\tilde u,v)\,d\tilde u
$$
and our bounds for $\nu$, we get
\begin{eqnarray*}
&&u_\lambda(v)+(r_+-\ckrp+o(1))c\,e^{-C_{\alpha,3}v}\\
&&\qquad\qquad\qquad\qquad\qquad\qquad\leq\ugrp(v)\leq\\
&&\qquad\qquad\qquad\qquad\qquad\qquad\qquad u_\lambda(v)+(r_+-\ckrp+o(1))Ce^{-c_{\alpha,3}v}.
\end{eqnarray*}
Taking into account our bounds~\eqref{u_lambda_ii} for $u_\lambda$, we obtain
\begin{equation}
\label{u_rp}
ce^{-C_{\alpha,3}v}\leq\ugrp(v)\leq Ce^{-c_{\alpha,3}v}.
\end{equation}
Thus,
\begin{equation}
\label{v_rp}
\frac{1}{C_{\alpha,3}}\ln\left(\frac{c}{u}\right)\leq\vgrp(u)\leq
\frac{1}{c_{\alpha,3}}\ln\left(\frac{C}{u}\right).
\end{equation}

From~\eqref{u_rp}, we see that for $(u,v) \in J^-(\Gamma_{\ckrp})$
\begin{equation}
\label{u_rp_v}
u\leq u_{\ckrp}(v)\leq Ce^{-(2\kk-\delta)v}.
\end{equation}

We can estimate $\nu$ over $\Gamma_{\ckrp}$ by
combining~\eqref{nu_2} with~\eqref{v_rp}:
\begin{equation}
\label{nu_rp}
-C\left(\frac{1}{u}\right)^{\frac{C_{\alpha,3}}{c_{\alpha,3}}}\leq
\nu(u,\vgrp(u))\leq-c\left(\frac{1}{u}\right)^{\frac{c_{\alpha,3}}{C_{\alpha,3}}}.
\end{equation}

From~\eqref{lambda_u} and~\eqref{kappa_u} we obtain
\begin{equation}\label{mu_u}
 \partial_u(1-\mu)=\partial_u\Bigl(\frac\lambda\kappa\Bigr)
 =\nu\partial_r(1-\mu)-(1-\mu)\frac{\nu}{r}\Bigl(\frac\zeta\nu\Bigr)^2.
\end{equation}
So, for points in $\{(u,v)\in J^+({\cal A}): v\geq V\}$, we have, taking into account the definition of $\varpi_m$, 
$$
\partial_u(1-\mu)\leq\nu\partial_r(1-\mu)\leq\nu\partial_r(1-\mu)(r,\varpi_m).
$$
It follows that, for $\ckrp=r_+-\delta$,
\begin{eqnarray}
(1-\mu)(u,v)&=&\int_{u_\lambda(v)}^{\ugrp(v)}\partial_u(1-\mu)(\tilde u,v)\,d\tilde u\nonumber\\
&\leq&\int_{r_+}^{\ckrp}\partial_r(1-\mu)(\tilde r,\varpi_0)\,d\tilde{r}-2\int^{\ckrp}_{r_+}
\frac{\varpi_0-\varpi_m}{\tilde r^2}\,d\tilde r\nonumber\\
&=&(1-\mu)(\ckrp,\varpi_0)+2(\varpi_0-\varpi_m)\frac{r_+-\ckrp}{r_+\ckrp}\nonumber\\ 
&\leq&-\left(\frac{2\kk}{1+\eps}
-\frac{2(\varpi_0-\varpi_m)}{r_+\ckrp}
\right)\delta,\label{negative}
\end{eqnarray}
where $0<\eps<1$ is fixed,
provided $\delta$ is sufficiently small. 

Suppose that $(u,v)\in \Gamma_{\ckrp}$. Integrating~\eqref{mu_u} yields
$$
(1-\mu)(u,v)=\int_{u_\lambda(v)}^{\ugrp(v)}e^{-\int_{\tilde u}^u\bigl(\frac{\nu}{r}\bigl(\frac\zeta\nu\bigr)^2\bigr)(\bar u,v)\,d\bar u}
\nu\partial_r(1-\mu)(\tilde u,v)\,d\tilde u.
$$
In this expression, we can use
$$
e^{-\int_{\tilde u}^u\bigl(\frac{\nu}{r}\bigl(\frac\zeta\nu\bigr)^2\bigr)(\bar u,v)\,d\bar u}
\leq\left(\frac{r_+}{\ckrp}\right)^{C_{\ckrp}^2}
$$
and
$$
\nu\partial_r(1-\mu)\geq\nu\partial_r(1-\mu)(r,\varpi_0)+\nu\frac{2(\varpi_M-\varpi_0)}{r^2}.
$$
For $\ckrp=r_+-\delta$, we then obtain
\begin{eqnarray}
&&(1-\mu)(u,v)\geq\nonumber\\
&&\qquad\left(\frac{r_+}{\ckrp}\right)^{C_{\ckrp}^2}
\left((1-\mu)(\ckrp,\varpi_0)-\,\frac{2(\varpi_M-\varpi_0)(r_+-\ckrp)}{r_+\ckrp}\right)\geq\nonumber\\
&&\qquad-\,\left(\frac{r_+}{\ckrp}\right)^{C_{\ckrp}^2}
\left(\frac{2\kk}{1-\eps}+\frac{2(\varpi_M-\varpi_0)}{r_+\ckrp}\right)\,\delta,\label{negative_below}
\end{eqnarray}
where $\eps$ is any fixed positive number,
provided $\delta$ is sufficiently small. 

Combining~\eqref{negative} and~\eqref{negative_below}
with~\eqref{kappa_2}, we have, for $(u,v)\in\Gamma_{\ckrp}$
\begin{eqnarray}
&&-\left(\frac{r_+}{\ckrp}\right)^{C_{\ckrp}^2}
\left(\frac{2\kk}{1-\eps}+\frac{2(\varpi_M-\varpi_0)}{r_+\ckrp}\right)\,\delta
\leq\lambda(u,v)\label{lambda_2}\\
&&\qquad\qquad\qquad\qquad
\leq-C_{\kappa,2}
\left(\frac{2\kk}{1+\eps}-\frac{2(\varpi_0-\varpi_m)}{r_+\ckrp}\right)\delta.
\nonumber
\end{eqnarray}

Inequality~\eqref{dlmz_2} and equation~\eqref{lambda_u} imply that $\partial_u\lambda<0$ in $J^-(\Gamma_{\ckrp})\cap J^+({\cal A})$.
Thus, an equality analogous to~\eqref{bar} together with~\eqref{lambda_2} yields
\begin{eqnarray*}
|\theta(u,v)-\theta(u_\lambda(v),v)|&\leq& |\lambda|(u,v)\max_{[u_\lambda(v),u]\times\{v\}}\Bigl|\frac{\zeta}{\nu}\Bigr|
\int_{u_\lambda(v)}^u\left[\frac{-\nu}{r}\right](\tilde u,v)\,d\tilde u\\
&\leq&C\ln\left({\frac{r_+}{\ckrp}}\right)
\max_{[u_\lambda(v),u]\times\{v\}}\Bigl|\frac{\zeta}{\nu}\Bigr|.
\end{eqnarray*}

When $\s>2$, using~\eqref{zn_ii_a},
$$
|\theta(u,v)-\theta(u_\lambda(v),v)|\leq C
e^{-(2\kk-\delta)v}.
$$
In this situation,
$$
2\kk<\frac{C_1}{2}=\s\kk-\eps
$$
(for sufficiently small $\eps$),
and so, in view of~\eqref{theta_l_ii},
\begin{equation}\label{theta_rp_ii_a}
|\theta|(u,v)\leq C
e^{-(2\kk-\delta)v}
\end{equation}
in $J^-(\Gamma_{\ckrp})\cap J^+({\cal A})$.

When  $\s<2$, using~\eqref{constant4_new},
\begin{eqnarray}
|\theta|(u,v)&\leq&
\label{theta_rp_ii_c} C
e^{-(\s\kk-\delta)v}.
\end{eqnarray}

Note that, as mentioned at the beginning of this section, the decay of $\theta$ has been overrun by that of $\frac{\zeta}{\nu}$.

\section{The region $J^-(\Gamma_{\ckrm})\cap J^+(\Gamma_{\ckrp})$}\label{section6}

In this section we focus on the region $J^-(\Gamma_{\ckrm})\cap J^+(\Gamma_{\ckrp})$, for a given ${\ckrm} \in \left]r_-,r_0\right[$, which will later be set conveniently close to $r_-$. From this point on our analysis will follow closely the methods used in~\cite{relIst2} and \cite{relIst3}. We will prove that the exponential decay for $\frac{\zeta}{\nu}$ and $\theta$, obtained in the previous section, persists in this new region. This is a consequence of the fact that the overall contribution of the redshift and blueshift effects is essentially neutral here, and so the decays carry over from $\Gamma_{\ckrp}$ to $\Gamma_{\ckrm}$. As in Section~5 of~\cite{relIst2}, the solution still behaves qualitatively as the Reissner-Nordstr\"om solution: $\varpi$ is close to $\varpi_0$, $\kappa$ is close to $1$ and $\zeta, \theta$ are close to $0$.

More precisely, we start by bounding $1-\mu$ from above by a negative constant,
and $\partial_r(1-\mu)$ from below. However,
we do not estimate the pair $\frac{\zeta}{\nu}$ and $\theta$,
as we did in the previous two sections, because we do not have $\partial_u\lambda<0$,
and so it is not easy to bound the double integral of $\frac{\nu\lambda}{r}$ appearing in the two-dimensional Gronwall's inequality.
Instead, we go on to
estimate the pair $\frac{\zeta}{\nu}$ and $\frac{\theta}{\lambda}$
as in \cite{relIst2}, using equation~(54) therein;
the bounds on $1-\mu$ and $\partial_r(1-\mu)$
allow us to obtain an upper bound for the exponentials in that formula.
Estimates for $\kappa$ from below and
$\varpi$ from above follow. Moreover, $1-\mu$ is clearly bounded from below.
Integrating the Raychaudhuri equations, and using the estimates for $\lambda$ and $\nu$
over $\Gamma_{\ckrp}$ together with the bounds on $1-\mu$, lead to
bounds for $\lambda$ and $\nu$. Finally, we obtain estimates for 
$v_{\ckrm}$ and $u_{\ckrm}$, which can be used to improve our previous estimates on
$\frac{\zeta}{\nu}$ and $\frac{\theta}{\lambda}$. The estimate for $\theta$ is essentially the same as the estimate for $\frac{\theta}{\lambda}$, as $\lambda$ is bounded.

In $J^+({\cal A})$, the mass is bounded below by $\varpi_m$ (recall~\eqref{omega_m}). We assume that
$\ckrm$ is sufficiently close to $r_-$ so that
$(1-\mu)(\ckrp,\varpi_0)\leq(1-\mu)(\ckrm,\varpi_0)$. Then,
for $\ckrm\leq r\leq\ckrp$,
\begin{eqnarray*}
(1-\mu)(r,\varpi_m)&=&(1-\mu)(r,\varpi_0)+\frac{2(\varpi_0-\varpi_m)}{r}\\
&\leq&(1-\mu)(\ckrm,\varpi_0)+\frac{2(\varpi_0-\varpi_m)}{\ckrm}\\
&=&(1-\mu)(\ckrm,\varpi_m).
\end{eqnarray*}
So,
in the region $J^-(\Gamma_{\ckrm})\cap J^+(\Gamma_{\ckrp})$, we have
\begin{equation}
(1-\mu)(r,\varpi)\ \leq\ (1-\mu)(r,\varpi_m)\ \leq\  (1-\mu)(\ckrm,\varpi_m)\ <\ 0,\label{sun}
\end{equation}
provided that $V$ is chosen sufficiently large for $\varpi_m$ to be close enough to
$\varpi_0$, so that $(1-\mu)(\ckrm,\varpi_m)<0$. The inequality $(1-\mu)(\ckrp,\varpi_m)\leq(1-\mu)(\ckrm,\varpi_m)$ follows from $(1-\mu)(\ckrp,\varpi_0)\leq(1-\mu)(\ckrm,\varpi_0)$ as above.
Moreover, for $\ckrm\leq r\leq\ckrp$,
\begin{eqnarray}
\partial_r(1-\mu)(r,\varpi)&=&\partial_r(1-\mu)(r,\varpi_0)-\,\frac{2(\varpi_0-\varpi)}{r^2}\nonumber\\
&\geq&\partial_r(1-\mu)(\ckrm,\varpi_0)-\,\frac{2(\varpi_0-\varpi_m)}{\ckrm^2}\nonumber\\
&\geq&\partial_r(1-\mu)(\ckrm,\varpi_m).\label{moon}
\end{eqnarray}
For the first inequality above, see the beginning of Section~3 in \cite{relIst2}.
Combining~\eqref{sun} with~\eqref{moon}, we have
in the region $J^-(\Gamma_{\ckrm})\cap J^+(\Gamma_{\ckrp})$
$$
\frac{\partial_r(1-\mu)}{(1-\mu)}\leq
\frac{\partial_r(1-\mu)(\ckrm,\varpi_m)}{(1-\mu)}
\leq\frac{\partial_r(1-\mu)(\ckrm,\varpi_m)}{(1-\mu)(\ckrm,\varpi_m)}
=:c_{\ckrm}.
$$
Thus, the exponentials in~(54) of \cite{relIst2} can be bounded in the following way:
\begin{equation}\label{c}
e^{-\int_{\vgrp(u)}^v[\kappa\partial_r(1-\mu)](u,\tilde v)\,d\tilde v}\leq e^{c_{\ckrm}(\ckrp-\ckrm)}=
:C.
\end{equation}

We will use Bondi coordinates $(r,v)$, where
\begin{equation}
(u,v)\mapsto(r(u,v),v)\quad \Leftrightarrow\quad (r,v)\mapsto (\ug(v),v).
\label{coordinates}
\end{equation}
We denote by $\hatz$ the function $\frac\zeta\nu$ written in Bondi coordinates, so that
$$
\frac\zeta\nu(u,v)=\hatz(r(u,v),v)
\quad \Leftrightarrow\quad
\hatz(r,v)=\frac\zeta\nu(\ug(v),v).
$$
The same notation will be used for other functions.

Let $r \in \left[\ckrm,\ckrp\right]$.
As in Section~5 of \cite{relIst2}, for $s \in \left[r,\ckrp\right]$, define
\begin{equation}\label{z}
 {\cal Z}_{(r,v)}(s)=\max_{\tilde v\in[\vgs(\ug(v)),v]}\Bigl|\hatz\Bigr|(s,\tilde v)
\end{equation}
and
\begin{equation}\label{t}
{\cal T}_{(r,v)}(\ckrp)=\max_{\tilde v\in[\vgrp(\ug(v)),v]}\Bigl|\widehat{\frac\theta\lambda}\Bigr|(\ckrp,\tilde v).
\end{equation}
Recall that at the beginning of Section~\ref{section5} we chose $U \leq u_\lambda(V)$; this guarantees that $\vgs(u) \geq V$ for all $s \leq \ckrp$ and $0 < u \leq U$, and so the quantities above are well defined.

\begin{center}
\begin{turn}{45}
\begin{psfrags}
\psfrag{p}{{\tiny $(u_r(v),v)$}}
\psfrag{j}{{\tiny $\Gamma_s$}}
\psfrag{g}{{\tiny $\Gamma_{\ckrp}$}}
\psfrag{h}{{\tiny ${\cal A}$}}
\psfrag{a}{{\tiny $v_s(u_r(v))$}}
\psfrag{b}{{\tiny $v$}}
\psfrag{c}{{\tiny $v_{\ckrp}(u_r(v))$}}
\psfrag{u}{{\tiny $v$}}
\psfrag{v}{{\tiny $u$}}
\psfrag{f}{{\tiny $u_\lambda(V)$}}
\psfrag{t}{{\tiny $U$}}
\psfrag{x}{{\tiny $u_r(v)$}}
\psfrag{y}{{\tiny $V$}}
\includegraphics[scale=1]{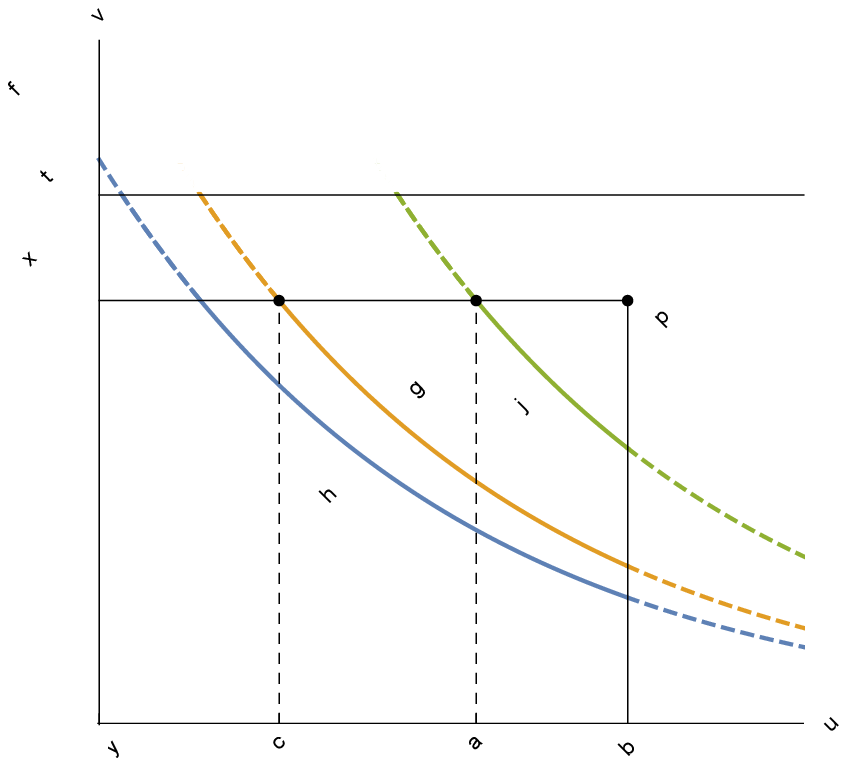}
\end{psfrags}
\end{turn}{45}
\end{center}

Let us define
\begin{equation}\label{ls}
l(s):=
\begin{cases}
\frac{\s}{2}&{\rm if}\ \s\leq 2 \\
1&{\rm if}\ \s>2.\\
\end{cases}
\end{equation}

From~\eqref{zn_ii_a} if $\s>2$, and~\eqref{constant4_new}
if $\s\leq 2$, we get
\begin{eqnarray}
\label{part-1a_ii_a}
{\cal Z}_{(r,v)}(\ckrp)&\leq& 
Ce^{-(2\kk l(s)-\delta)v_{\ckrp}(u_r(v))}.
\end{eqnarray}
Similarly, using \eqref{theta_rp_ii_a} if $\s>2$, 
and \eqref{theta_rp_ii_c} if $\s \leq 2$, together with \eqref{lambda_2}, we obtain
\begin{eqnarray}
\label{part-2a_ii_a}
{\cal T}_{(r,v)}(\ckrp)&\leq& C
e^{-(2\kk l(s)-\delta)v_{\ckrp}(u_r(v))}.
\end{eqnarray}
Arguing as in the proof of Lemma~5.1 of \cite{relIst2} leads to
\begin{equation}\label{zf}
{\cal Z}_{(r,v)}(r)\leq C\Bigl[{\cal Z}_{(r,v)}(\ckrp)+C\ln\Bigl(\frac{\ckrp}{r}\Bigr){\cal T}_{(r,v)}(\ckrp)\Bigr]e^{\frac{C^2(\ckrp-r)^2}{r\ckrp}}.
\end{equation}
Here the constant $C$ is as in~\eqref{c}.
Substituting~\eqref{part-1a_ii_a} and~\eqref{part-2a_ii_a} in~\eqref{zf} finally yields the key estimate
\begin{equation}\label{z-0_ii_a}
\Bigl|\hatz\Bigr|(r,v)\leq 
Ce^{-(2\kk l(s)-\delta)v_{\ckrp}(u_r(v))},
\end{equation}
for $\ckrm\leq r\leq\ckrp$.

From~\eqref{z-0_ii_a} we can now estimate the remaining quantities. Continuing to argue as in the proof of Lemma~5.1 of \cite{relIst2}, we have, using~\eqref{c},
\begin{eqnarray}
\Bigl|\widehat{\frac\theta\lambda}\Bigr|(r,v)&\leq&C\Bigl|\widehat{\frac\theta\lambda}\Bigr|(\ckrp,v)
+C\int_{r}^{\ckrp}\Bigl[\Bigl|\hatz\Bigr|\frac 1 {\tilde s}\Bigr](\tilde s,v)\,d\tilde s,\label{zfq}
\end{eqnarray}
again with the constant $C$ as in~\eqref{c}.
We then use~\eqref{part-2a_ii_a} and~\eqref{z-0_ii_a} in~\eqref{zfq}
to obtain
\begin{equation}\label{t-0_ii_a}
\Bigl|\widehat{\frac\theta\lambda}\Bigr|(r,v)\leq 
Ce^{-(2\kk l(s)-\delta)v_{\ckrp}(u_r(v))},
\end{equation}
for $\ckrm\leq r\leq\ckrp$.

Inequalities~\eqref{z-0_ii_a} and \eqref{t-0_ii_a},
together with~\eqref{v_rp}, show that, given $\hat{\delta}>0$, we can choose
$U$ sufficiently small so that $\bigl|\frac{\zeta}{\nu}\bigr|<\hat{\delta}$ and $\bigl|\frac{\theta}{\lambda}\bigr|<\hat{\delta}$
in the region $J^-(\Gamma_{\ckrm})\cap J^+(\Gamma_{\ckrp})$. Arguing as 
in~\eqref{kappa_2}, we conclude that
\begin{equation}\label{kappa_3}
C_{\kappa,3}:=C_{\kappa,2}\left(\frac{\ckrm}{\ckrp}\right)^{\hat{\delta}^2}\leq\kappa\leq 1,
\end{equation}
in $J^-(\Gamma_{\ckrm})\cap J^+(\Gamma_{\ckrp})$.

A version of~\eqref{omega_final} together with \eqref{omega_m} and~\eqref{omega_M} imply that, for $(u,v)\in
J^-(\Gamma_{\ckrm})\cap J^+(\Gamma_{\ckrp})$,
\begin{eqnarray*}
\varpi_0+o(1)= \varpi_m\leq& &\\
\varpi(u_{\ckrp}(v),v)\leq&\varpi(u,v)&\leq\varpi(u_{\ckrp}(v),v)+C\hat{\delta}^2
\\
&&\leq\varpi_M+C\hat{\delta}^2\\
&&\leq\varpi_0+o(1)+C\hat{\delta}^2.
\end{eqnarray*}

The proof of Lemma~5.2 of \cite{relIst2} shows that the curve $\Gamma_{\ckrm}$
intersects every line of constant $u$, and so 
$\lim_{v\to\infty}u_{\ckrm}(v)=0$.
In particular, the inequalities above
imply that
\begin{equation}\label{lim_varpi}
\lim_{\stackrel{(u,v)\in{J^-(\Gamma_{\ckrm})}}{{\mbox{\tiny{$v\to\infty$}}}}}\varpi(u,v)=\varpi_0.
\end{equation}

We rewrite~\eqref{kappa_3} in the form
$$
C_{\kappa,3}\leq\frac{\lambda}{1-\mu}(u,v)\leq 1.
$$
We have bounded $1-\mu$ from above by a negative constant in~\eqref{sun}.
Since $\varpi$ is bounded above, and $r$ is bounded below by a positive constant, $1-\mu$ is also
bounded from below, and so there exist positive
constants $\tilde c$ and $C$ such that
\begin{equation}\label{tilde_c}
-C\leq\lambda\leq -\tilde c
\end{equation}
$J^-(\Gamma_{\ckrm})\cap J^+(\Gamma_{\ckrp})$.

Integrating the Raychaudhuri equation~\eqref{ray_v_bis}
and taking into account the estimate
\begin{equation}\label{perto_um}
\Bigl(\frac{\ckrm}{r_+}\Bigr)^{\hat\delta^2}\leq e^{\int_{v_{\ckrp}(u)}^v\bigl((\frac{\theta}{\lambda})^2\frac{\lambda}{r}\bigr)(u,\tilde v)\,d\tilde v}\leq 1
\end{equation}
(which uses $\bigl|\frac{\theta}{\lambda}\bigr|<\hat{\delta}$), we deduce
\begin{equation}\label{nu_mu}
\Bigl(\frac{\ckrm}{r_+}\Bigr)^{\hat\delta^2}\frac{\nu}{1-\mu}(u,v_{\ckrp}(u))
\leq
\frac{\nu}{1-\mu}(u,v)\leq\frac{\nu}{1-\mu}(u,v_{\ckrp}(u)).
\end{equation}
Hence, combining the estimate~\eqref{nu_rp} with our bounds for $1-\mu$ in the region
$J^-(\Gamma_{\ckrm})\cap J^+(\Gamma_{\ckrp})$,
we arrive at
\begin{equation}
\label{nu_rm}
-C\left(\frac{1}{u}\right)^{\frac{C_{\alpha,3}}{c_{\alpha,3}}}\leq
\nu(u,v)\leq-c\left(\frac{1}{u}\right)^{\frac{c_{\alpha,3}}{C_{\alpha,3}}},
\end{equation}
for $(u,v)\in J^-(\Gamma_{\ckrm})\cap J^+(\Gamma_{\ckrp})$.

Integrating~\eqref{tilde_c} between $v_{\ckrp}(u)$ and $v$, for 
$(u,v)\in J^-(\Gamma_{\ckrm})\cap J^+(\Gamma_{\ckrp})$, we obtain
\begin{equation}\label{diff}
v-v_{\ckrp}(u)\leq{\textstyle\frac{\ckrp-r}{\tilde c}}=:
 c_{r,\ckrp}\leq c_{\ckrm,\ckrp}.
\end{equation}
Taking into account~\eqref{v_rp}, we have
\begin{equation}
\label{v_rm}
\frac{1}{C_{\alpha,3}}\ln\left(\frac{c}{u}\right)\leq v_{\ckrm}(u)\leq
c_{\ckrm,\ckrp}+\frac{1}{c_{\alpha,3}}\ln\left(\frac{C}{u}\right).
\end{equation}
This yields
\begin{eqnarray}
\label{u_rm}
ce^{-C_{\alpha,3}v}\ \leq\ u_{\ckrm}(v)&\leq& Ce^{c_{\alpha,3}c_{\ckrm,\ckrp}}
e^{-c_{\alpha,3}v}\\
&=&Ce^{-c_{\alpha,3}v}.\nonumber
\end{eqnarray}
Notice that this constant $C$ blows up as $\ckrp\nearrow r_+$ and $\ckrm\searrow r_-$, because $\tilde c$ approaches zero.

Using the estimate~\eqref{diff}, we can obtain improved estimates for
$\bigl|\frac{\zeta}{\nu}\bigr|$ and $\bigl|\frac{\theta}{\lambda}\bigr|$
in $J^-(\Gamma_{\ckrm})\cap J^+(\Gamma_{\ckrp})$, given by
\begin{equation}\label{zt-0_ii_a}
\Bigl|\frac\theta\lambda\Bigr|(u,v)+\Bigl|\frac{\zeta}{\nu}\Bigr|(u,v)\leq 
Ce^{-(2\kk l(s)-\delta)v}.
\end{equation}

From \eqref{tilde_c} we also conclude that
\begin{equation}
|\theta|(u,v)\leq 
Ce^{-(2\kk l(s)-\delta)v}. \label{theta_no_shift}
\end{equation}

\section{The region $J^-(\gamma)\cap J^+(\Gamma_{\ckrm})$}\label{section7}

As in Section~6 of \cite{relIst2}, we define a spacelike curve $\gamma=\gamma_{\ckrm,\beta}$ to the future of
$\Gamma_{\ckrm}$, parameterized by
\begin{equation}\label{small_gamma}
u\mapsto\big(u,(1+\beta)\,\vgr(u))=:(u,v_\gamma(u))
\end{equation}
for $u\in[0,U]$, where
\begin{equation}\label{beta_ii} 
0<\beta<{\textstyle\frac 12\left(\sqrt{1+8l(s)\frac{\kk}{\kkk}}-1\right)}.
\end{equation}
Here $k_- := \frac12 \left|\partial_r(1-\mu)(r_-,\varpi_0)\right|$ denotes the surface gravity of the Cauchy horizon for the Reissner-Nordstr\"{o}m black hole with parameters $r_-$ and $\varpi_0$.
Unlike the boundaries of the regions studied in the previous sections, this curve is not a level set of the radius function. Its purpose is to probe the geometry of the region where the blueshift effect, which is dominant at the Cauchy horizon, starts being felt. This is characterized by an exponential growth of the form $e^{2k_- v}$
Nevertheless, to the past of $\gamma$ the function $r$ is bounded below and the mass $\varpi$ is bounded above, so that
the solution still behaves qualitatively as in the interior of the Reissner-Nordstr\"om black hole. More precisely, we have
\begin{Lem}\label{boot_r}
For each $\beta$ as above, there exist $\overline{\ckrm} \in \left]r_-, r_0\right[$ and $\eps_0\in\left]0,r_-\right[$ for which, whenever\/ $\ckrm$ and $\eps$ are chosen satisfying $\ckrm\in\left]r_-,\overline{\ckrm}\right]$ and $\eps\in\left]0,\eps_0\right]$, the following holds:
there exists $U_\eps$ (depending on $\ckrm$ and $\eps$) such that if $(u,v)\in J^-(\gam)\cap J^+(\cg_{\ckrm})$,
with $0<u\leq U_\eps$,
then 
\begin{equation}\label{r_omega_2}
r(u,v)\geq r_--{\textstyle\frac\eps 2}\quad{\rm and}\quad\varpi(u,v)\leq\varpi_0+{\textstyle\frac\eps 2}.
\end{equation}
\end{Lem}
\begin{proof}
Let $(u,v)\in J^-(\gamma)\cap J^+(\Gamma_{\ckrm})$ be such that
$r(u,v)\geq r_--\eps\geq r_--\eps_0$.
Recall from the proof of Lemma~6.1 of \cite{relIst2} that for 
$(u,v)\in J^-(\gamma)\cap J^+(\Gamma_{\ckrm})$ there exists a constant $\underline C$ 
(depending on $r_--\eps_0$) such that 
\begin{eqnarray}
&&\int_{\vgr(u)}^v|\theta|(u,\tilde v)\,d\tilde v+\int_{\ugr(v)}^u|\zeta|(\tilde u,v)\,d\tilde u\label{Big}\\
&&\qquad\leq\underline C\left(
\int_{\vgr(u)}^v|\theta|(\ugr(v),\tilde v)\,d\tilde v+\int_{\ugr(v)}^u|\zeta|(\tilde u,\vgr(u))\,d\tilde u
\right).\nonumber
\end{eqnarray}
Following the proof of Lemma~6.1 of \cite{relIst2},
we see that the crucial step is to bound the integral
$$
\int_{\ugr(v)}^u\Bigl[\Bigl|\frac\zeta\nu\Bigr||\zeta|\Bigr](\tilde u,v)\,d\tilde u,
$$
for $(u,v)\in J^-(\gamma)\cap J^+(\Gamma_{\ckrm})$,
by a function that goes to zero when $v$ goes to infinity. In order to do that, 
we bound the first integral on the right-hand side of~\eqref{Big} by 
using the estimates for $\theta$ obtained above.
These are \eqref{theta_l_ii} in ${\cal P}_\lambda$;
\eqref{theta_rp_ii_a}$-$\eqref{theta_rp_ii_c} in $J^-(\Gamma_{\ckrp})\cap J^+({\cal A})$;
and \eqref{theta_no_shift} in $J^-(\Gamma_{\ckrm})\cap J^+(\Gamma_{\ckrp})$.
In summary, the upper bound \eqref{theta_no_shift} can be used in $J^-(\Gamma_{\ckrm})$.
Therefore, using  $v-\vgr(u)\leq \beta\vgr(u)$ and $\vgr(u)=\frac{\vgam(u)}{1+\beta}\geq\frac{v}{1+\beta}$,
we have
\begin{eqnarray*}
\int_{\vgr(u)}^v|\theta|(\ugr(v),\tilde v)\,d\tilde v
&\leq&Ce^{-\left(2\kk l(s)-\delta\right)\vgr(u)}\beta\vgr(u)\\
&\leq&Ce^{-\left(\frac{2\kk l(s)}{1+\beta}-\delta\right)v}.
\end{eqnarray*}

To bound the second integral on the right-hand side of~\eqref{Big},
we use the estimates for $\frac{\zeta}{\nu}$ obtained above.
These are \eqref{constant} in ${\cal P}_\lambda$;
\eqref{zn_ii_a} and~\eqref{constant4_new}
in $J^-(\Gamma_{\ckrp})\cap J^+({\cal A})$;
and \eqref{zt-0_ii_a} in $J^-(\Gamma_{\ckrm})\cap J^+(\Gamma_{\ckrp})$. 
Thus, we can write
\begin{eqnarray*}
\int_{\ugr(v)}^u|\zeta|(\tilde u,\vgr(u))\,d\tilde u
&=&\int_{\ugr(v)}^u\left[\left|\frac{\zeta}{\nu}\right|(-\nu)\right](\tilde u,\vgr(u))\,d\tilde u\\
&\leq&C(\ckrm-r(u,v))e^{-(2\kk l(s)-\delta)\vgr(u)}\\
&\leq&Ce^{-\left(\frac{2\kk l(s)}{1+\beta}-\delta\right)v}.
\end{eqnarray*}

It follows that the left-hand side of~\eqref{Big} can be bounded by
\begin{eqnarray}
\int_{\vgr(u)}^v|\theta|(u,\tilde v)\,d\tilde v+\int_{\ugr(v)}^u|\zeta|(\tilde u,v)\,d\tilde u\leq
&Ce^{-\left(\frac{2\kk l(s)}{1+\beta}-\delta\right)v}.\label{int_theta_zeta}
\end{eqnarray}

In $J^+(\Gamma_{\ckrm})$ we have $\varpi\geq\varpi_m$ (see \eqref{omega_m}). So, for $(u,v)\in J^-(\gamma)\cap J^+(\Gamma_{\ckrm})$ with $r(u,v)\geq r_--\eps_0$,
as in~\eqref{moon},

\begin{equation}\label{desigual}
\partial_r(1-\mu)\geq \partial_r(1-\mu)(r_--\eps_0,\varpi_m),
\end{equation}
whence, using $\kappa \leq 1$,
\begin{eqnarray*}
e^{-\int_{\vgr(u)}^{v}[\kappa\partial_r(1-\mu)](u,\tilde v)\,d\tilde v}&\leq&
e^{-\partial_r(1-\mu)(r_--\eps_0,\varpi_m)\beta\vgr(u)}\\ &\leq& e^{-\partial_r(1-\mu)(r_--\eps_0,\varpi_m)\beta v}.
\end{eqnarray*}
Thus, integrating \eqref{zeta_nu} from $\Gamma_{\ckrm}$ (similar to \eqref{rafaeli}) we have
\begin{eqnarray} 
\nonumber
\Bigl|\frac\zeta\nu\Bigr|(u,v)&\leq&\Bigl|\frac\zeta\nu\Bigr|(u,\vgr(u))
e^{-\int_{\vgr(u)}^{v}[\kappa\partial_r(1-\mu)](u,\tilde v)\,d\tilde v}
\\
\nonumber
&&+\int_{\vgr(u)}^{v}\frac{|\theta|} r(u,\bar v)e^{-\int_{\bar v}^{v}[\kappa
\partial_r(1-\mu)](u,\tilde{v})\,d\tilde{v}}\,d\bar v
\\
\nonumber
&\leq&Ce^{-(2\kk l(s)-\delta)\vgr(u)}e^{-\partial_r(1-\mu)(r_--\eps_0,\varpi_m)\beta v}
\\
\nonumber
&&+\frac{e^{-\partial_r(1-\mu)(r_--\eps_0,\varpi_m)\beta v}}{r_--\eps_0}\int_{\vgr(u)}^{v}|\theta|(u,\bar v)\,d\bar v
\\
\nonumber
&\leq&Ce^{-\bigl(\frac{2\kk l(s)}{1+\beta}-\delta\bigr) v}e^{-\partial_r(1-\mu)(r_--\eps_0,\varpi_m)\beta v}
\\
\nonumber
&&+\frac{e^{-\partial_r(1-\mu)(r_--\eps_0,\varpi_m)\beta v}}{r_--\eps_0}Ce^{-\bigl(\frac{2\kk l(s)}{1+\beta}-\delta\bigr)v}
\\
&\leq&Ce^{-\bigl(\frac{2\kk l(s)}{1+\beta}+\partial_r(1-\mu)(r_--\eps_0,\varpi_m)\beta-\delta\bigr) v}.\label{zetaNuGamma}
\end{eqnarray}

From the previous estimate and \eqref{int_theta_zeta} we obtain
\begin{eqnarray}
&&\int_{\ugr(v)}^u\Bigl[\Bigl|\frac\zeta\nu\Bigr||\zeta|\Bigr](\tilde u,v)\,d\tilde u\nonumber\\ 
&&\ \ \ \ \leq
Ce^{-\bigl(\frac{2\kk l(s)}{1+\beta}+\partial_r(1-\mu)(r_--\eps_0,\varpi_m)\beta-\delta\bigr) v}
\int_{\ugr(v)}^u|\zeta|(\tilde u,v)\,d\tilde u\nonumber\\
&&\ \ \ \ \leq Ce^{-\bigl(\frac{4\kk l(s)}{1+\beta}+\partial_r(1-\mu)(r_--\eps_0,\varpi_m)\beta-\delta\bigr) v}.\label{K}
\end{eqnarray}

The constant in the exponent,
\begin{equation}\label{exponent}
\frac{4\kk l(s)}{1+\beta}+\partial_r(1-\mu)(r_--\eps_0,\varpi_m)\beta-\delta,
\end{equation}
is positive for 
\begin{equation}\label{success} 
\textstyle
\beta<\frac 12\left(
\sqrt{(1+\tilde\delta)^2-8
\frac{2\kk l(s)-\frac{\delta}{2}}{\partial_r(1-\mu)(r_--\eps_0,\varpi_m)}}-(1+\tilde\delta)
\right),
\end{equation}
where
$$
\textstyle
\tilde{\delta}=-\,\frac{\delta}{\partial_r(1-\mu)(r_--\eps_0,\varpi_m)}.
$$
The right-hand side tends to
$$
{\textstyle\frac 12\left(\sqrt{1+8l(s)\frac{\kk}{\kkk}}-1\right)}
$$
as $(\ckrp,\eps_0,\delta,\varpi_m)\to(r_+,0,0,\varpi_0)$.
So, if $\beta$ satisfies~\eqref{beta_ii}, we may choose $\ckrp$
sufficiently close to $r_+$, $\eps_0$ and $\delta$
sufficiently small, and
$U$ sufficiently small (which in $J^-(\gamma)$ implies $v$ sufficiently large,
so that $\varpi_m$
is sufficiently close to $\varpi_0$) so that~\eqref{success} holds.

Now that we have the bound \eqref{K} with the right-hand side going to zero as $v \to \infty$,
the formula
\begin{eqnarray*}
\varpi(u,v)&\leq&
\varpi(\ugr(v),v)e^{\frac{1}{r_--\eps_0}\int_{\ugr(v)}^u\bigl[\bigl|\frac\zeta\nu\bigr||\zeta|\bigr](\tilde u,v)\,d\tilde u}\\
&&+C 
\int_{\ugr(v)}^ue^{\frac{1}{r_--\eps_0}\int_{s}^u\bigl[\bigl|\frac\zeta\nu\bigr||\zeta|\bigr](\tilde u,v)\,d\tilde u}
\Bigl[\Bigl|\frac\zeta\nu\Bigr||\zeta|\Bigr](s,v)\,ds
\end{eqnarray*}
and the fact that $\lim_{v\to\infty}\varpi(\ugr(v),v)=\varpi_0$ (recall~\eqref{lim_varpi}) imply that for each $0 <\bar\eps < \eps_0$ there exists $\bar U_{\bar\eps} > 0$ such that
\begin{eqnarray*}
\varpi(u,v)
&\leq&\varpi_0+{\textstyle\frac{\bar\eps} 2},
\end{eqnarray*}
provided that  $u\leq\bar U_{\bar\eps}$.
Since $1-\mu$
is nonpositive in $J^+({\cal A})$ and $1-\mu=(1-\mu)(r,\varpi_0)-\frac{2(\varpi-\varpi_0)}{r}$, we have
$$\textstyle
(1-\mu)(r(u,v),\varpi_0)\leq\frac{2(\varpi(u,v)-\varpi_0)}{r}\leq\frac{\bar\eps}{r_--\eps_0}.
$$
Hence, by inspection of the graph of $(1-\mu)(r,\varpi_0)$, there exists $\bar\eps_0$ such that for $0<\bar\eps\leq\bar\eps_0$, we have $r(u,v)>r_--\frac\eps 2$ provided that  $u\leq\bar U_{\bar\eps}$.
For $0<u\leq U_\eps:=\min\{\bar U_{\bar\eps_0},\bar U_\eps\}$, both inequalities~\eqref{r_omega_2} hold. A standard bootstrap argument now yields the result, as the sets
\[
\{ (u,v)\in J^-(\gamma)\cap J^+(\Gamma_{\ckrm}) : r(u,v) > r_--\eps \}
\]
and
\[
\{ (u,v)\in J^-(\gamma)\cap J^+(\Gamma_{\ckrm}) : r(u,v) \geq r_--{\textstyle\frac\eps 2} \},
\]
coincide, and are therefore both open and closed in the relative topology of the connected set $J^-(\gamma)\cap J^+(\Gamma_{\ckrm})$.
\end{proof}

From the previous proof it is clear that,
given $\eps>0$, we may choose $U$ sufficiently small so that
if $(u,v)\in J^-(\gamma)\cap J^+(\Gamma_{\ckrm})$, then
\begin{equation}\label{H} 
1-\eps \leq \kappa(u,v)\leq 1.
\end{equation}

Now we turn to the behavior of $\lambda$ and $\nu$ over the curve $\gamma$.
The conclusions of Lemma~6.6 of \cite{relIst2} still hold in our case:
\begin{Lem}\label{lambda-and-nu}
Suppose that $\beta$ is given satisfying~\eqref{beta_ii}.
Let $\gam$ be the curve parametrized by~\eqref{small_gamma}.
Let also $\delta>0$. For a choice of $\ckrm$ sufficiently close to $r_-$, and $U$
sufficiently small,
there exist constants $c$ and $C$, such that
 for  $(u,v)\in\gam$, with $0<u\leq U$, we have
 \begin{equation} 
 ce^{\bigl(-2\kkk\frac\beta{1+\beta}-\delta\bigr)v}\leq -\lambda(u,v)\leq Ce^{\bigl(-2\kkk\frac\beta{1+\beta}+\delta\bigr)v}\label{lambda-above}
 \end{equation}
 and
  \begin{equation}
    cu^{\mbox{\tiny$\,\frac{\kkk}{\kk}\beta$}\,-1+\delta} \leq -\nu(u,v)\leq Cu^{\mbox{\tiny$\,\frac{\kkk}{\kk}\beta$}\,-1-\delta}.
  \label{nu-above}
  \end{equation}
\end{Lem}
\begin{proof}
The proof of Lemma~6.6 of \cite{relIst2}
goes through with minor modifications. Hence, we point out that formula~(126) of
\cite{relIst2} should be replaced by
\begin{eqnarray} 
\nonumber
&&\left|\frac\theta\lambda\right|(u,v)\\ \nonumber
&&
\ \leq \left(Ce^{-(2\kk l(s)-\delta) v}+
Ce^{-\bigl(\frac{2\kk l(s)}{1+\beta}-\delta\bigr) v}\right)\times
\\
\nonumber
&&\qquad\qquad\qquad\qquad\qquad\qquad\times
e^{-\partial_r(1-\mu)(r_--\eps_0,\varpi_m)\mbox{\tiny$\frac{\min_{\cg_{\ckrm}}(1-\mu)}{\max_{\cg_{\ckrm}}(1-\mu)}$}\beta v}\\
&&\ \leq Ce^{-\left(\frac{2\kk l(s)}{1+\beta}+
\partial_r(1-\mu)(r_--\eps_0,\varpi_m)\mbox{\tiny$\frac{\min_{\cg_{\ckrm}}(1-\mu)}{\max_{\cg_{\ckrm}}(1-\mu)}$}\beta-\delta\right)v}.\label{thetaLambdaGamma}
\end{eqnarray}
This leads to (127) of \cite{relIst2}.
\end{proof}
Let us also point out that, analogously to (135) and (136) in \cite{relIst2}, given $\delta>0$,
\begin{equation} 
ce^{(-2\kk-\delta)\frac{v_\gamma(u)}{1+\beta}}\label{O_ii} \leq u \leq Ce^{(-2\kk+\delta)\frac{v_\gamma(u)}{1+\beta}},
\end{equation}
for $u\leq U$ sufficiently small.

\bigskip

The inequalities \eqref{lambda-above} highlight the importance of the curve $\gamma$ in probing the geometry of the region near the Cauchy horizon. These exponential decays, which will be crucial to establish the integrability of $\lambda$, and consequently the stability of the radius function at the Cauchy horizon, already exhibit the characteristic blueshift exponent $-2k_-$, multiplied by the positive parameter $\beta$. Observe that these estimates cannot be obtained over level curves of $r$ (corresponding to $\beta=0$).

\section{The region $J^+(\gamma)$}\label{section8}

In this section we treat the region $J^+(\gamma)$, where the solution departs qualitatively from the Reissner-Nordstr\"{o}m solution. Nevertheless the radius function remains bounded away from zero, and approaches $r_-$ as $u \to 0$, implying that the existence of a Cauchy horizon is a stable property.

We may apply the arguments in the proof of Lemma~7.1 of~\cite{relIst2}.
We see that the estimates~\eqref{lambda-above} and~\eqref{nu-above},
 \begin{eqnarray} 
  -\lambda(u,v)\leq
  Ce^{\bigl(-2\kkk \frac\beta{1+\beta}+\delta\bigr)v}\label{lambda-above_bis}
 \end{eqnarray}
 and
  \begin{eqnarray}
  -\nu(u,v)\leq
 Cu^{\mbox{\tiny$\,\frac{\kkk}{\kk}\beta$}\,-1-\delta},
  \label{nu-above_bis}
  \end{eqnarray}
also hold in $\{ r > r_- - \eps \}\cap J^+(\gamma)$ for $\eps>0$ sufficiently small. 

Using the integrability of $\lambda$ and $\nu$ implicit in~\eqref{lambda-above_bis} or~\eqref{nu-above_bis}, as in Section~7 of~\cite{relIst2},
we can prove the stability of the Cauchy horizon.

\begin{Thm}\label{stability_of_Cauchy_horizon}
Given $\delta>0$, there exists $U_\delta>0$ such that $r(u,v)>r_--\delta$ for $(u,v)\in J^+(\gamma)$
with $u\leq U_\delta$. In particular, ${\cal P}$ contains
$[0,U_\delta]\times[0,\infty[$. 
\end{Thm}

Due to the monotonicity properties of $r$ and $\varpi$, the limits $r(u,\infty)=\lim_{v\to\infty}r(u,v)$ and $\varpi(u,\infty)=\lim_{v\to\infty}\varpi(u,v)$ are
well defined, and 
$$\lim_{u\searrow\, 0}r(u,\infty)=r_-.$$

The proofs of Theorem~8.1 and Lemma~8.2 of~\cite{relIst2} establish
\begin{Lem}\label{rmenos}
Either $r(\,\cdot\,,\infty)\equiv r_-$ and $\varpi(\,\cdot\,,\infty)\equiv \varpi_0$, or
$r(u,\infty)<r_-$ and $\varpi(u,\infty)>\varpi_0$ for all $u>0$. In the second
case $\int_0^\infty\kappa(u,v)\,dv<\infty$ and
$\liminf_{v\to\infty}-\nu(u,v)>0$, for all $u>0$.
\end{Lem}

We also have
\begin{Lem}\label{affine}
Let $u>0$. Consider an outgoing null geodesic $t\mapsto(u,v(t))$ for $({\cal M},g)$, with $g$ given by
$$
g=-\Omega^2(u,v)\,dudv+r^2(u,v)\,\sigma_{{\mathbb S}^2},
$$
where
$\Omega^2=-4\kappa\nu$.
Then $v^{-1}(\infty)<\infty$, i.e.\
the affine parameter is finite at the Cauchy horizon.
\end{Lem}
\begin{proof}
Let $u>0$. Fix a $V>v_\lambda(u)$ such that $(1-\mu)(u,V)<0$.
As shown in the proof of Corollary~8.3 of~\cite{relIst2}, there exists a constant $c>0$
such that
$$
t=v^{-1}(V)+c\int_V^v\Omega^2(u,\bar v)\,d\bar v=v^{-1}(V)-4c\int_V^v(\nu\kappa)(u,\bar v)\,d\bar v.
$$
Integrating~\eqref{ray_v_bis}, for $\bar v\geq V$, we get
$$ 
0<\frac{\nu(u,\bar v)}{(1-\mu)(u,\bar v)}\leq\frac{\nu(u,V)}{(1-\mu)(u,V)}.
$$ 
So
\begin{eqnarray*}
t&\leq& v^{-1}(V)-4c\frac{\nu(u,V)}{(1-\mu)(u,V)}\int_V^v\lambda(u,\bar v)\,d\bar v\\
&\leq& v^{-1}(V)+4c\frac{\nu(u,V)}{(1-\mu)(u,V)}r_+<\infty.
\end{eqnarray*}
\end{proof}

\section{Mass inflation}\label{section9}

As mentioned in page~\pageref{regime}, there exist two distinct regimes, depending on the parameter $\s > 0$, which reflect the competition phenomenon between the redshift arising from the evolution equation and the exponential decay of $\theta_0$ along the event horizon: for $\s < 2$ the decay of $\theta_0$ is slower, and thus the dominant effect, whereas for $\s > 2$ this decay is overwhelmed by the redshift effect. Since mass inflation is more likely for slower decays, it is not surprising that sufficient conditions for its occurrence can be obtained when $\s < 2$.

\subsection{Positivity of $\theta$ and $\zeta$}

In this subsection we prove positivity of $\theta$ and $\zeta$
over ${\cal A}$ (for large $v$). This implies positivity of $\theta$ and $\zeta$
in $J^+({\cal A})$, which in turn imply $\varpi(u,\infty)>\varpi_0$ for all $u>0$.

Note that if $(r,\nu,\lambda,\varpi,\theta,\zeta,\kappa)$ is a solution of the first order
system \eqref{r_u}$-$\eqref{kappa_at_u}, then 
$(r,\nu,\lambda,\varpi,-\theta,-\zeta,\kappa)$ is also a solution of that system.
So, without loss of generality, taking into account \eqref{theta_l_ii}, we assume that
\begin{equation}\label{theta_l_ii_plus}
\theta(u,v)\geq ce^{-\,\frac{C_2}{2}v},
\end{equation}
for $(u,v)\in{\cal P}_\lambda$.

According to~\eqref{c-alpha}, for $(u,v)\in{\cal R}_{(U,V)}$ and $V\leq\tilde v\leq v$,
\begin{equation}\label{bounds}
e^{-C_\alpha(v-\tilde v)}\leq
e^{-\int_{\tilde v}^{v}[\kappa\partial_r(1-\mu)](u,\bar{v})\,d\bar{v}}\leq 
e^{-c_\alpha(v-\tilde v)}.
\end{equation}
Integrating~\eqref{zeta_nu}, we obtain, analogously to~\eqref{rafaeli},
\begin{eqnarray}
\frac\zeta\nu(u,v)&=&\frac\zeta\nu(u,V)e^{-\int_{V}^{v}[\kappa\partial_r(1-\mu)](u,\tilde v)\,d\tilde v}\nonumber\\
&&-\int_{V}^{v}\frac\theta r(u,\tilde v)e^{-\int_{\tilde v}^{v}[\kappa\partial_r(1-\mu)](u,\bar{v})\,d\bar{v}}\,d\tilde v.\label{field_7}
\end{eqnarray}
Using \eqref{theta_l_ii_plus} and~\eqref{bounds} in~\eqref{field_7} yields, for $(u,v)\in{\cal R}_{(U,V)}$,
$$
\frac\zeta\nu(u,v)\leq Ce^{-c_\alpha(v-V)}-\frac{c}{r_+}\int_{V}^{v}e^{-\,\frac{C_2}{2}\tilde v}
e^{-C_\alpha(v-\tilde v)}\,d\tilde v.
$$
Assuming that $\s>2$ we can choose our parameters so that
$\frac{C_2}{2}<c_\alpha<C_\alpha$. Then
$$
\frac\zeta\nu(u,v)\leq Ce^{-c_\alpha(v-V)}-ce^{-\,\frac{C_2}{2}v}+
ce^{-C_\alpha(v-V)}e^{-\,\frac{C_2}{2}V}.
$$
This shows there exists $\bar V\geq V$ such that
$$
\frac\zeta\nu(u_\lambda(v),v)<0\ \ \ \mbox{for}\ v\geq \bar V,
$$
and so
\begin{equation}
\label{zeta_l_ii_plus}
\zeta(u_\lambda(v),v)>0\ \ \ \mbox{for}\ v\geq \bar V.
\end{equation}

We restrict $U$ to be at most $u_\lambda(\bar V)$.
With this choice, \eqref{theta_l_ii_plus} and~\eqref{zeta_l_ii_plus} ensure that
$\theta$ and $\zeta$ are positive on ${\cal A}$. 
This implies that $\theta>0$ and $\zeta>0$ in $J^+({\cal A})$:
otherwise, there would exist a point $(u,v)\in J^+({\cal A})$ such that $\theta(u,v)=0$ or $\zeta(u,v)=0$
but $\theta>0$ and $\zeta>0$ in $J^-(u,v)\cap J^+({\cal A})$. Integrating~\eqref{theta_u} and~\eqref{zeta_v},
we would obtain a contradiction.

From Corollary~12.3 of~\cite{Dafermos2}, it follows that, for $0<u_1<u_2$ and $v$ sufficiently large
so that $(u_1,v)$ (and hence $(u_2,v)$) belong to $J^+({\cal A})$,
$$
\varpi(u_2,\infty)-\varpi(u_1,\infty)\geq \varpi(u_2,v)-\varpi(u_1,v).
$$
The right-hand side is positive because $\zeta$ is positive on $J^+({\cal A})$. We conclude
that $\varpi(u,\infty)>\varpi_0$ for all $u>0$. Moreover, Lemma~\ref{rmenos} implies $r(u,\infty)<
r_-$ for all $u>0$.

\subsection{Blow up of the mass at the Cauchy horizon}

We define 
\begin{equation}\label{def_psi}
\Psi:=\frac{k_-}{k_+} > 1.
\end{equation}

The following result establishes sufficient conditions for the occurrence of mass inflation.

\begin{Thm}\label{mass_inflation_thm}
If $\s<\min\left\{\rho,2\right\}$ then $\varpi(u,\infty)=\infty$ for all $u>0$. 
\end{Thm}

\begin{proof}
We proved in the previous subsection that 
$\varpi(u,\infty)>\varpi_0$ for all $u>0$ and so it follows from Lemma~\ref{rmenos}
that $r(u,\infty)<r_-$ for all $u>0$. Going through the proof of 
Theorem~3.1 in~\cite{relIst3}, we see that to prove mass inflation it is sufficient to
consider Case~3.2, namely it is sufficient to assume that
$$
I(u)=\int_{\vgr(u)}^\infty\left[\frac{\theta^2}{-\lambda}\right](u,\tilde v)\,d\tilde v
$$
satisfies $\lim_{u\searrow 0}I(u)=0$ and from this derive the contradiction
$I(u)=\infty$. This is done by using improved upper bounds for $-\lambda$
in the region $J^+(\gamma)$ together with the lower bounds satisfied by $\theta$ in
this region.

The assumption $\lim_{u\searrow 0}I(u)=0$ together with~\eqref{lambda-above} leads
to~(117) of~\cite{relIst3}, which states that
\begin{eqnarray}
-\lambda(u,v)\leq C(u)
e^{(-2\kkk+\delta)v}\label{lambda_cima_ii}
\end{eqnarray}
in $J^+(\gamma)$.
When $\s<2$, we know that the lower bound for $\theta$ 
in~\eqref{theta_l_ii_plus} holds on ${\cal A}$
and we know that $\theta$ and $\zeta$ are positive on $J^+({\cal A})$. Hence,
the lower bound for $\theta$ 
in~\eqref{theta_l_ii_plus} holds on $J^+({\cal A})$.
Using~\eqref{theta_l_ii_plus} and~\eqref{lambda_cima_ii}, we are then led to the following lower bound for 
$I(u)$:
\begin{eqnarray*}
I(u)&\geq&\int_{\vgam(u)}^\infty\left[\frac{\theta^2}{-\lambda}\right](u,\tilde v)\,d\tilde v\\
&\geq&C(u)\int_{\vgam(u)}^\infty\frac{e^{-C_2\tilde v}}{e^{(-2\kkk+\delta)\tilde v}}\,d\tilde v.
\end{eqnarray*}
We can choose our parameters so that this integral is infinite if (see~\eqref{C2})
$$\textstyle
\frac{2}{r_+}(A-r_+\kk)<2\kkk.
$$
This inequality is equivalent to 
$$
\s<\rho.
$$
Therefore
$\varpi(u,\infty)=\infty$ for all $u>0$ if $s<\min\left\{\rho,2\right\}$.
\end{proof}

\subsection{Mass inflation or $\frac{\theta}{\lambda}$ unbounded}

Suppose that $\s<2$ and that $\varpi(\,\cdot\,,\infty)$ is
not identically equal to $\infty$. Then, from the proof of Theorem~3.2 of~\cite{relIst3} we know
that $\lim_{u\searrow 0}I(u)=0$, that $-\lambda$ is bounded above by~\eqref{lambda_cima_ii}
in $J^+(\gamma)$, and that $\theta$ is bounded below by~\eqref{theta_l_ii_plus}.
We conclude that, for $(u,v)\in J^+(\gamma)$,
\begin{eqnarray}
\Bigl|\frac{\theta}{\lambda}\Bigr|(u,v)&\geq&
\frac{e^{-\,\frac{C_2}{2}v}}
{C(u)e^{(-2\kk+\delta)v}}
\nonumber\\
&=& \frac 1{C(u)}e^{\bigl(-\,\frac{1}{r_+}(A-r_+\kk)
+2\kkk-\delta\bigr)v}.\label{penelope}
\end{eqnarray}
This exponent can be made positive if
$$\textstyle
\frac{1}{r_+}(A-r_+\kk)<2\kkk,
$$
which is equivalent to 
$$
\s<2\rho.
$$
However, as shown in Appendix~A of~\cite{relIst3}, $\rho$ is necessarily greater than one,
and so the last inequality is always satisfied for $\s<2$. Therefore we have the following result.  

\begin{Thm}
If $\s<2$ then either $\varpi$ or $\frac{\theta}{\lambda}$ blow up at the Cauchy horizon.
\end{Thm}

\section{No mass inflation}\label{section10}

In this section we will establish sufficient conditions guaranteeing that the renormalized Hawking mass does not blow up at the Cauchy horizon. As might be expected, this is what occurs in the regime $\s > 2$, where the decay of the initial data is faster, if the reference Reissner-N\"ordstrom black hole is sufficiently close to extremality. More surprising is the fact that mass inflation can also be avoided for $\s < 2$, although $\bigl|\frac{\theta}{\lambda}\bigr|$ necessarily blows up.

\begin{Thm}\label{no_mass}
Suppose that\/ $\rho$ satisfies $\frac{7\rho}{9}<l(s)$ (see~\eqref{ls}), that is,
\begin{equation}\label{ls2}
\left\{
\begin{array}{ll}
1<\rho<\frac{9}{7}&{\rm if}\ \s>2,\\
\vspace{-3mm}\\
1<\rho<\frac{9}{14}\s
&{\rm if}\ \frac{9}{14}<s\leq 2.
\end{array}
\right.
\end{equation}
Then $\varpi(u,\infty)<\infty$ for each $0<u\leq U$, provided that $U$ is sufficiently small. Furthermore,\/ $\lim_{u\searrow 0}\varpi(u,\infty)=\varpi_0$.
\end{Thm}
\begin{proof}
Given $\eps_1>0$, define
\begin{eqnarray*}
{\cal D}=
{\cal D}_{\eps_1}=\left\{(u,v)\in J^+(\gamma):\ u\leq U\/\ {\rm and}\   
\int_{\vgam(u)}^v\Bigl| \frac{\theta^2}{\lambda} \Bigr|(u,\tilde v)\,d\tilde v\leq \eps_1\right\}.
\end{eqnarray*}
The conclusion follows by proving that $\cD$ is open in $J^+(\gamma)$ if $\eps_1$ and $U$ are sufficiently small.
As in the proof of Theorem~4.1 in~\cite{relIst3}, this is accomplished 
by deriving a formula showing that in $\cD$ we have
\begin{equation}\label{exponent_delta}
\Bigl|\frac{\theta^2}{\lambda}\Bigr|(u,v)\leq Ce^{-\Delta v}
\end{equation}
for $0<u\leq U$.

From~\eqref{thetaLambdaGamma} we get
\begin{equation} 
\Bigl|\frac\theta\lambda\Bigr|(u_\gamma(v),v)
\leq Ce^{-\left(\frac{2\kk l(s)}{1+\beta}-
2\kkk \beta-\delta\right)v},\label{piece1}
\end{equation}
and recall that in~\eqref{lambda-above} we obtained
\begin{equation} 
 ce^{\bigl(-2\kkk\frac\beta{1+\beta}+\delta\bigr)v}\leq
  -\lambda(u_\gamma(v),v)\leq
  Ce^{\bigl(-2\kkk\frac\beta{1+\beta}-\delta\bigr)v}.\label{piece2}
 \end{equation}
 Combining~\eqref{piece1} with~\eqref{piece2} yields
 \begin{eqnarray}
 |\theta|(\ugam(v),v)&\leq& Ce^{-\left(\frac{2\kk l(s)}{1+\beta}-\frac{2\kkk\beta^2}{1+\beta}-\delta\right)v}.\label{bbtheta}
 \end{eqnarray}
Moreover, according to~\eqref{O_ii},
\begin{equation*} 
\frac{1+\beta}{2\kk+\delta}\ln\Bigr(\frac cu\Bigr) \leq \vgam(u)\leq
\frac{1+\beta}{2\kk-\delta}\ln\Bigr(\frac Cu\Bigr).
\end{equation*}
Therefore, the proof of Lemma~4.2 of~\cite{relIst3} goes through if one replaces
$s+1$ by $l(s)$. For example, for $(u,v)\in{\cal D}$, $q=\frac{1}{3}$ and $\beta=\frac{1}{3}+\eps$
with $\eps$ sufficiently small, we have
\begin{eqnarray}
|\theta|(u,v)&\leq& Ce^{-2\left(\frac{k_+l(s)}{1+\beta}-\,\frac{k_-\beta^2}{1+\beta}-\delta\right)v}\nonumber\\
&&+Cu^{l(s)-\Psi(\beta^2+q)-\delta}e^{-2\left(\frac{k_-(\beta+q)}{1+\beta}-\delta\right)v}.\label{eq10}
\end{eqnarray}
An estimate of the form~\eqref{exponent_delta} then follows if we assume 
$$
l(s)>\Psi(\beta^2+\beta+q)>\frac{7\rho}{9}.
$$
\end{proof}

\section{Breakdown of the Christodoulou-Chru\'sciel criterion}\label{section11}

In this section we prove that when there is no mass inflation the solution can be extended across the Cauchy horizon with enough regularity to violate the Christodoulou-Chru\'sciel version of strong cosmic censorship. The extension is constructed by first changing $v$ to a new coordinate with finite range, essentially the distance to the apparent horizon as measured by the radius function along $u=U$. This is the most natural choice to bring the Cauchy horizon to a finite coordinate value. If there is no mass inflation, all functions except $\theta$ are then shown to extend continuously to any subset of the Cauchy horizon away from the event horizon. We then change to yet another coordinate system, where we are able to prove that the Christoffel symbols are locally square integrable.

We regard the $(u,v)$ plane, the domain of our first order system, as a $C^2$ manifold.
We define a new null coordinate along the outgoing direction by
\begin{equation}\label{change}
\tilde v=r(U,V_\lambda)-r(U,v),
\end{equation}
where $V_\lambda=\max\{v:\lambda(U,v)=0\}$.
Equality~\eqref{omega_eh} and the assumption that $\hat{f}$ is continuous and integrable
imply that $\hat{\varpi}$ is a continuously differentiable function of $r$.
Similarly, \eqref{kappa_eh} and~\eqref{lambda_eh} guarantee that $\hat\kappa$ and 
$\hat\lambda$ are also continuously differentiable functions of $r$.
Then, equation~\eqref{r-v} shows that the coordinate $v$ over the event horizon is a
continuously differentiable function of $r$ with nonvanishing derivative, so that,
by the Inverse Function Theorem, $r$ is a continuously differentiable function
of~$v$ over the event horizon. We conclude that $\kappa$, $\varpi$ and $\lambda$ are
continuously differentiable functions of $v$ over the event horizon (in particular, $r$ is a $C^2$ function of $v$ along the event horizon). In addition, $\nu_0$
is continuously differentiable. Therefore, hypothesis (h4) in Section~6 of~\cite{relIst1} is 
satisfied, and so, by Lemma~6.1 in~\cite{relIst1}, the function $r$ is $C^2$.
Since we have $\lambda(U,v)<0$ for $v>V_\lambda$, equation~\eqref{change} allows us to define an
 admissible coordinate change
 $$[0,U]\times\left]V_\lambda,\infty\right[\,\ \rightarrow\ [0,U]\times\left]0,\tilde V\right[,$$
 where $\tilde V=r(U,V_\lambda)-r(U,\infty)$. We write a tilde over
 a function to indicate that we are using these new coordinates.

Assume that the hypotheses of Theorem~\ref{no_mass} hold. Let $0<\delta<U$.
As in Proposition~5.2 of~\cite{relIst3},
we can extend $\tilde r$ and $\tilde{\varpi}$ to continuous functions
on $[\delta,U]\times[0,\tilde V]$.

Using~\eqref{ray_v_bis} and the bound~\eqref{exponent_delta}
(which holds in $[\delta,U]\times[v,\infty[$, for $v>v_\gamma(\delta)$),
one proves that $\frac{\tilde\nu}{1-\tilde\mu}(\,\cdot\,,\tilde v)$
converges uniformly for $u\in[\delta,U]$ when $\tilde v\to \tilde V$.
As in Step 2 of the proof of Proposition~5.2 of~\cite{relIst3}, this implies that
$\frac{\tilde\nu}{1-\tilde\mu}$ admits a continuous extension to the
rectangle $[\delta,U]\times[0, \tilde V]$ (for arbitrary $\delta$). Notice that the function
$\frac{\tilde\nu}{1-\tilde\mu}(\,\cdot\,, \tilde V)$ is strictly positive on $]0,U]$.
When $\tilde r(u,\tilde V)=r_-$ and $\tilde\varpi(u,\tilde V)=\varpi_0$, we have $(1-\tilde\mu)(u,\tilde V)=0$; 
therefore, $\tilde\nu(u,\tilde V)$ exists and
is zero. On the other hand, when $\tilde r(u,\tilde V)<r_-$, Lemma~\ref{rmenos} implies that $\liminf_{\tilde v\to \tilde V}-\tilde{\nu}(u,\tilde v)>0$,
and so $(1-\tilde{\mu})(u,\tilde V)<0$. We conclude that in this case $\tilde{\nu}(u,\tilde V)$ exists and is negative.

The function $\tilde\lambda:=\partial_{\tilde v} r$ satisfies
$$
\tilde{\lambda}(U,\tilde v)\equiv -1.
$$
The integration of~\eqref{lambda_u} leads to
$$
\tilde \lambda(u,\tilde v)=\tilde\lambda(U,\tilde v)e^{-\int_u^U\bigl[\frac{\tilde\nu}{1-\tilde\mu}
\partial_{\tilde r}(1-\tilde\mu)\bigr](\tilde u,\tilde v)\,d\tilde u}.
$$
Since $\tilde\lambda(U,\tilde v)$ extends to $[0,\tilde V]$, and
$\frac{\tilde\nu}{1-\tilde\mu}$ and $\partial_{\tilde r}(1-\tilde\mu)$ extend to
 $[\delta,U]\times[0,\tilde V]$, $\tilde\lambda$ extends as a continuous function to $[\delta,U]\times[0,\tilde V]$.
 Moreover, $\tilde\lambda(\,\cdot\,,\tilde V)$ is strictly negative on $]0,U]$.

Taking into account the behavior of $\tilde\lambda$ and $\frac{\tilde\nu}{1-\tilde\mu}$
on $]0,U]\times\{\tilde V\}$, we see that the coefficient of the metric 
$$
\tilde{\Omega}^2=-4\tilde\kappa\tilde\nu=-4\tilde\lambda\frac{\tilde\nu}{1-\tilde\mu}
$$
is strictly positive on $\left]0,U\right]\times\{\tilde V\}$ and is continuous on $[\delta,U]\times[0,\tilde V]$,
for any $0<\delta<U$.

Equation~\eqref{nu_v} can be written as 
$$
\partial_{\tilde v}\tilde{\nu}=-\,\frac{\tilde{\Omega}^2}{4}\partial_{\tilde r}(1-\tilde{\mu}).
$$
Therefore the convergence of $\tilde{\nu}(\,\cdot\,,\tilde v)$ to $\tilde{\nu}(\,\cdot\,,\tilde V)$
is uniform for $u \in [\delta,U]$, and so $\tilde{\nu}$ is continuous on
 $[\delta,U]\times[0,\tilde V]$, for any $0<\delta<U$.
 
 Integrating~\eqref{zeta_v},
 $$ 
 \tilde\zeta(u,\tilde V)=\tilde\zeta(u,\tilde v)-\int_{\tilde v}^{\tilde V}\frac{\tilde\theta\tilde\nu}{\tilde r}(u,\bar v)\,d\bar v.
 $$ 
 We use
 $$\int_{\tilde v}^{\tilde V}|\tilde\theta|(u,\bar v)\,d\bar v=\int_{f^{-1}(\tilde v)}^{\infty}|\theta|(u,\bar v)\,d\bar v\to 0$$
 as $\tilde v\nearrow \tilde V$ (by~\eqref{eq10}). Note that the last convergence is uniform for $u\in[\delta,U]$.
 We may define $\tilde\zeta(\,\cdot\,,\tilde V)$ 
 as the uniform limit of $\tilde\zeta(\,\cdot\,,\tilde v)$ when $\tilde v\nearrow \tilde V$.
 
Therefore we have proved the following result.
 
\begin{Thm}\label{fechado-0}
Assume that the hypotheses of\/ {\rm Theorem~\ref{no_mass}} hold.
Then, for all\/ $0<\delta<U$, the functions $\tilde r$, $\tilde\nu$, $\tilde\lambda$, $\tilde\varpi$,
$\tilde\zeta$ and $\tilde\kappa$ (but not necessarily $\tilde{\theta}$)
admit continuous extensions to the closed rectangle\/ $[\delta,U]\times[0,\tilde V]$.
Moreover, ${(1-\tilde\mu)}(u,\tilde V)$ is negative for $u>0$, unless $r(\,\cdot\,,\infty)\equiv r_-$.
\end{Thm}

\begin{Rmk}
Since $\tilde r$ and $\tilde{\Omega}^2$ are strictly positive on $\left]0,U\right]\times\{\tilde V\}$, it is immediate to construct continuous extensions of the metric beyond the Cauchy horizon, as was done in the proof of\/ {\rm Corollary~5.11} of\/~{\rm \cite{relIst3}}.
\end{Rmk}

To construct extensions which also have locally square integrable Christoffel symbols it is useful to consider a new $\mathring v$ coordinate determined by the condition
\begin{equation}
\mathring \Omega^2(U,\mathring v)\equiv 1\;. 
\end{equation}
One has 
$$\frac{d \mathring v}{d v}=\Omega^2(U,v)$$
and consequently 
$$\frac{d \mathring v}{d \tilde v}=-\frac{\Omega^2(U,v)}{\lambda(U,v)}=\frac{4\nu}{1-\mu}(U,v)=\frac{4\tilde \nu}{1-\tilde \mu}(U,\tilde v)=-\frac{{\tilde\Omega}^2(U,\tilde v)}{\tilde\lambda(U,\tilde v)}\;.$$
We conclude that these coordinate systems are $C^1$-compatible up to and including $\tilde v=\tilde V$. In particular this shows that $\mathring V=\mv(\tilde V)$ is finite and that we can construct continuous extensions of the corresponding metric components beyond the Cauchy horizon $\mathring v=\mathring V$.  
 
Note that the choice of coordinates provided by $\hat \kappa (U,\hat v)\equiv 1$, used in \cite{relIst3} for an analogous extension, is not regular in the case when $\tilde r(u,\tilde V)\equiv r_-$, since $\frac{d\hat v}{d \tilde v}=-\frac1{1-\mu}$ diverges as we approach the Cauchy horizon.

Since
$$ \frac{\tilde \nu}{1-\tilde \mu}\equiv -\frac{{\tilde\Omega}^2}{4\tilde\lambda}$$
extends continuously to the Cauchy horizon, where it is strictly positive, the only potentially problematic Christoffel symbols are 
$$\mrpp \Gamma^u_{uu}=\partial_u \log \mrpp \Omega^2$$
and
$$\mrpp \Gamma^{\mv}_{\mv\mv}=\partial_{\mv} \log \mrpp \Omega^2 \;.$$
Now, in view of the boundedness of the quantities $\mathring r$, $\mathring\nu$, $\mathring\lambda$, $\mathring\varpi$, $\mathring\zeta$ and $\mathring\kappa$ (but not necessarily $\mathring{\theta}$) guaranteed by Theorem~\ref{fechado-0} and the $C^1$-compatibility of the two coordinate systems, Einstein's equation \eqref{wave_Omega} gives us, for $\mv<\mrpp V$, 
\begin{equation}
\label{Einstein}
 \partial_u\partial_{\mv}  \log \mrpp \Omega^2 = O(1) (\mrpp \theta +1)\;.
\end{equation}
Since our choice of coordinates gives $ \log \mrpp \Omega^2 (U,\mv)\equiv 0$, integrating the previous equation first in $u$ and then in $\mv$, while applying H\"older's inequality in between, gives
$$\int_{\mrpp v_0}^{\mrpp V} \left(\mrpp \Gamma^{\mv}_{\mv\mv}\right)^2(u,\mv) d\mv \leq 
C \left(1+\int_u^U \int_{\mrpp v_0}^{\mrpp V}\mrpp \theta^2(\bar u,\mv) d\mv  d\bar u \right)\;.$$       
The proof that leads to inequality \eqref{exponent_delta} also shows that 
\begin{equation}
\label{est}
\left| \frac{ \theta^2(u,v)}{\lambda(U,v)}\right|\leq C e^{-\Delta v}\;,
\end{equation}
since it uses an upper bound for $|\theta|$ and a lower bound for $|\lambda|$, both of which are uniform in $u$. Therefore, 
$$\int_{\mrpp v_0}^{\mrpp V}\mrpp \theta^2(u,\mv) d\mv = \int_{v_0}^{\infty}\theta^2(u,v)\frac{1}{\Omega^2(U,v)} dv \leq C  \int_{v_0}^{\infty}\theta^2(u,v)\frac{1}{-\lambda(U,v)} dv\leq C\;,$$
and so
$$\int_{\mrpp v_0}^{\mrpp V} \left(\mrpp \Gamma^{\mv}_{\mv\mv}\right)^2(u,\mv) d\mv \leq C.$$
  
Note that smoothness provides $|\partial_u\log\mrpp \Omega^2(u,\mv_0)|\leq C$ for fixed $\mv_0$. Then, integrating~\eqref{Einstein} in $v$ leads to
$$\left|\mrpp \Gamma^u_{uu}\right|(u,\mv) \leq C\left(1+\int_{\mv_0}^{\mv} |\mrpp \theta|(u,\mv) d\mv \right)\leq C\;,$$
again by H\"older's inequality.

Therefore we have the following result.

\begin{Thm}\label{breakdown} 
Let ${\cal M}_\delta$ be the preimage of $[\delta,U]\times[0,\mathring{V}]$ by the null coordinate functions $(u,\mathring v)$. Then the Christoffel symbols and $\mrpp\theta$ are in $L^2({\cal M}_\delta)$.
\end{Thm}

\begin{proof}
The square of the $L^2$ norm of a function $\mrpp h$ on ${\cal M}_\delta$ is given by
$$
\int_{{\cal M}_\delta}{\mrpp h}^2\,dV_4=4\pi\int_{[\delta,U]\times[0,\mrpp{V}]}\left[\mrpp r^2\frac{\mrpp\Omega^2}2\mrpp h^2\right](u,\mrpp v)\,dud\mrpp v.
$$
Since the functions $\mrpp r$ and $\mrpp\Omega^2=-4\mrpp\nu\mrpp\kappa$ are bounded in $[\delta,U]\times[0,\mrpp{V}]$, we conclude that the Christoffel symbols and $\mrpp\theta$ are in $L^2({\cal M}_\delta)$.
\end{proof}

\begin{Rmk}
Again, as was done in the proof of
{\rm Corollary~5.11} of\/~{\rm \cite{relIst3}}, it is easy to construct extensions of the metric beyond the Cauchy horizon whose Christoffel symbols are in $L^2_{{\rm loc}}$ (and whose scalar field is in $H^1_{{\rm loc}}$). In other words, the Christodoulou-Chru\'sciel version of strong cosmic censorship does not hold in this setup.
\end{Rmk}

\end{document}